%% file: holes.tex
\def\showuniversal{\yes}

\documentclass[a4paper]{article}

\usepackage{microtype} 

\usepackage{placeins}

\usepackage{fullpage}

\usepackage{float}
\restylefloat{figure}

\usepackage[mathscr]{euscript}
\usepackage{amssymb}

\usepackage{subfigure}

\usepackage{amsmath}
\usepackage{amsthm}
\usepackage[rgb]{xcolor}

\usepackage{tikz}
\usetikzlibrary{calc,decorations.pathreplacing}

\usepackage{enumitem}  
\usepackage{units}

\usepackage{hyperref}

\hypersetup{%
  breaklinks,%
  ocgcolorlinks,
  colorlinks=true,%
  linkcolor=[rgb]{0.45,0.0,0.0},%
  citecolor=[rgb]{0,0,0.45}
}

\setlist{nolistsep}

\newtheorem{theorem}{Theorem}
\newtheorem{lemma}[theorem]{Lemma}
\newtheorem{corollary}[theorem]{Corollary}
\newtheorem{claim}[theorem]{Claim}

\newcommand{\email}[1]{#1}

\newcommand{\mb}{\mathbb}
\newcommand{\mc}{\mathscr}
\newcommand{\ns}{\scriptsize}

\newcommand{\minksum}{+}
\newcommand{\R}{\mathbb{R}}
\newcommand{\s}{\mathbb{S}}
\newcommand{\N}{\mathbb{N}}
\newcommand{\SC}{\Delta}
\newcommand{\im}{\textrm{im}\,}
\newcommand{\eps}{\varepsilon}
\newcommand{\hpaths}{\xi}
\newcommand{\tpaths}{\gamma}

\newcommand{\Back}{\textsc{Back}}
\newcommand{\Front}{\textsc{Front}}

\newcommand{\KK}{\Gamma}
\newcommand{\F}{\mc{F}}

\newcommand{\uz}{u_0}
\newcommand{\ut}{u_1}

\newcommand{\half}{\nicefrac{1}{2}}

\renewcommand{\vec}[1]{#1}

\DeclareMathOperator{\conv}{conv}
\DeclareMathOperator{\rank}{rank}
\DeclareMathOperator{\nerv}{\mathcal{N}}
\DeclareMathOperator{\chain}{\mathcal{C}}
\DeclareMathOperator{\nullity}{nullity}

\let\leq\leqslant
\let\geq\geqslant
\let\le\leqslant

\makeatletter
\partopsep\z@ \textfloatsep 10pt plus 1pt minus 4pt
\def\section{\@startsection {section}{1}{\z@}%
  {-3.5ex plus -1ex
    minus -.2ex}{2.3ex plus .2ex}{\large\bf}}
\def\subsection{\@startsection{subsection}{2}%
  {\z@}{-3.25ex plus
    -1ex minus -.2ex}{1.5ex plus .2ex}{\normalsize\bf}}
\def\@fnsymbol#1{\ensuremath{\ifcase#1\or \or 1\or 2\or 3\or 4\or
    5\or 6\or 7 \or 8\ or 9 \or 10\or 11 \else\@ctrerr\fi}}
\makeatother

\newcommand{\basiscoords}{
\path 
(-160:2 and 1) coordinate (e1)
(-70:2 and 1) coordinate (e2)
(90:1.73) coordinate (e3)
;
}
\newenvironment{yxzcoords}{
\basiscoords
\begin{scope}[x={(e2)}, y={(e1)}, z={(e3)}]
}{
\end{scope}
}

\newenvironment{yzxcoords}{
\basiscoords
\begin{scope}[x={(e2)}, y={(e3)}, z={(e1)}]
}{
\end{scope}
}

\newcommand{\polytopeB}[1]{
\draw[#1]
(0,0,2) coordinate (a1) -- (2,0,2) coordinate (a2) -- (0,-2,2) coordinate (a3) -- cycle
(0,0,0) coordinate (b1) -- (2,0,0) coordinate (b2) -- (0,-2,0) coordinate (b3) -- cycle
(a1) -- (1,1,1) coordinate (v) -- (b1) -- cycle
(a2) -- (v) -- (b2) -- cycle
(a3) -- (b3)
;
}

\begin{document}

\title{The Number of Holes in the Union of Translates \\ 
  of a Convex Set in Three Dimensions
  \thanks{B.~Aronov is supported by NSF
    grants CCF-11-17336 and CCF-12-18791.
    O.~Cheong and M.~G.~Dobbins are supported by NRF grant 2011-0030044
    (SRC-GAIA) from the government of Korea.}
} 

\author{Boris Aronov%
  \thanks{Polytechnic School of Engineering, NYU, Brooklyn, New York, USA. 
    \email{boris.aronov@nyu.edu}}
  \and 
  Otfried Cheong%
  \thanks{KAIST, Daejeon, South Korea.
    \email{otfried@kaist.edu}}
  \and
  Michael Gene Dobbins%
  \thanks{Department of Mathematical Sciences, Binghamton University,
      Binghamton, NY, USA.
    \email{michaelgenedobbins@gmail.com}}
  \and 
  Xavier Goaoc%
  \thanks{Universit\'e Paris-Est Marne-la-Vall\'ee,  France.
    \email{goaoc@upem.fr}}}

\maketitle

\begin{abstract}
  We show that the union of $n$ translates of a convex body 
  in~$\R^3$ can have~$\Theta(n^3)$ holes in the worst case, where a
  \emph{hole} in a set~$X$ is a connected component of $\R^3 \setminus
  X$. This refutes a 20-year-old conjecture. As a consequence, we
  also obtain improved lower bounds on the complexity of motion
  planning problems and of Voronoi diagrams with convex distance
  functions.
\end{abstract}

\section{Introduction}

From path planning in robotics~\cite{mp} to the design of
epsilon-nets~\cite{eps-union} to analyzing vulnerabilities in
networks~\cite{AHKS2014Union}, a variety of combinatorial and
algorithmic problems in computational geometry involve understanding
the complexity of the union of $n$ elementary objects.  An abundant
literature studies how this union complexity depends on the geometry
of the objects, and we refer the interested reader to the survey of
Agarwal, Pach and Sharir~\cite{StateOfUnion}.  In the plane, two
important types of conditions were shown to imply near-linear union
complexity: restrictions on the number of boundary
intersections~\cite{DS,UnionOfJordanRegions} and fatness
\cite{fat-convex,fat-triangles}.  In three
dimensions, the former are less relevant as they do not apply to
important examples such as motion planning problems.  
Whether fatness implies low
union complexity in $\R^3$ has been identified as an important open
problem in the area~\cite[Problem 4]{CGC42}; to quote Agarwal, Pach
and Sharir~\cite[Section 3.1, \S 2]{StateOfUnion},
\begin{quote}
  ``\emph{A prevailing conjecture is that the maximum complexity of the union
  of such fat objects is indeed at most nearly quadratic. Such a bound
  has however proved quite elusive to obtain for general fat objects,
  and this has been recognized as one of the major open problems in
  computational geometry.}''
\end{quote}

A weaker version of this conjecture asserts that the number of holes
in the union of $n$ translates of a fixed convex body in $\R^3$ is at most
nearly quadratic in $n$. By a \emph{hole} in a subset $X \subseteq
\R^d$ we mean a connected component of $\R^d \setminus X$. This is
indeed a weakening since the number of holes is a lower bound on the
complexity and a family of translates can be made quite fat by
applying a suitable affine transformation. Evidence in support of the
weaker conjecture is that in two dimensions, the union of $n$
translates has at most a linear number of
holes~\cite{UnionOfJordanRegions}, and in three dimensions, the union
of $n$ translates of a convex polytope with $k$ facets has at most
$O(k n^2)$ holes~\cite{OnTranslationalMotionPlanning}, so the number
of holes grows only quadratically for any fixed convex polytope.

Remarkably, we refute the conjecture even in this weaker form. 
We construct a convex body in $\R^3$ that has, for
any~$n$, a family of $n$ translates with $\Theta(n^{3})$ holes in its
union.  This matches the upper bound for
families of arbitrary convex bodies in~$\R^{3}$.

%
%
%
\begin{theorem}
  \label{thrm:cubic voids}
  The maximum number of holes in the union of $n$~translates of a
  compact, convex body in $\R^3$ is~$\Theta(n^3)$.
\end{theorem}

We  start with  a warm-up  example  illustrating the  idea behind  our
construction    (Section~\ref{s:ex1}).     We   then    construct    a
polytope~$K_{m}$,  tailored to  the value  of~$m$ considered  and with
$\Theta(m^2)$ faces;  we give a  family of $3m$~translates  of $K_{m}$
(in  Section~\ref{s:ex2}) whose  union has  $\Theta(m^3)$ holes.   The
final   step  of   our  construction   is  to   turn  the   family  of
polytopes~$(K_{m})$  into  a  limiting ``universal''  convex  body~$K$
that,   for   any~$m$,   admits  $3m$~translates   whose   union   has
$\Theta(m^3)$~holes.   We prove  this  formally  using arguments  from
algebraic topology, developed in Section~\ref{sec:nerve}.

\paragraph*{Further consequences.}

We conclude this introduction with two examples of problems in
computational geometry whose underlying structure involves a union of
translates of a convex body, and on which Theorem~\ref{thrm:cubic
  voids} casts some new light.

A \emph{motion-planning problem} asks whether an object, typically in
$\R^2$ or $\R^3$, can move from an initial position to a final
position by a sequence of elementary motions while remaining disjoint
from a set of obstacles (and to compute such a motion when it exists).
This amounts to asking whether two given points lie in the same
connected component of the \emph{free space}; that is, the set of
positions of the object where it intersects no obstacle. When the
motions are restricted to translations, the free space can be obtained
by taking the complement of the union of the ``expansion'' (formally:
the \emph{Minkowski sum}) of every obstacle by the reflection of the
object through the origin.  In the simplest case the mobile object is
convex, the obstacles consist of~$n$ points and the free space is the
complement of the union of $n$ translates of a convex body;
Theorem~\ref{thrm:cubic voids} implies that already in this case the
free space can have large complexity:

\begin{corollary}
  There exists a set $P$ of $n$ point obstacles and a convex body $K$
  in~$\R^3$ such that the free space for moving $K$ by translations
  while avoiding $P$ has~$\Theta(n^3)$ connected components.
\end{corollary}

The \emph{Voronoi diagram} of a family $P$ of points $p_1, p_2,
\ldots, p_n$, called \emph{sites}, in a metric space $X$ is the
partition of $X$ according to the closest $p_i$. A subset $Q \subseteq
P$ \emph{defines a face} of the diagram if there exist some point $x
\in X$ at equal distance from all sites of $Q$, and strictly further
away from all sites in $P\setminus Q$.  A case of interest is when $X$
is $\R^d$ equipped with a \emph{convex distance function} $d_U$
defined by a convex unit ball~$U$ (in general, $d_U$ is not a metric,
as $U$ needs not be centrally symmetric).  A \emph{face} of the
Voronoi diagram with respect to $d_U$ is defined to be a
connected component of the set $\{ x : d(x,q) \leq d(x,p), \text{ for
  all } q \in Q, p \in P \}$ for some $Q \subset P$ defining a face; the
\emph{complexity} of a Voronoi diagram is measured by the
number of its faces of all dimensions.  In $\R^2$, the complexity of
the Voronoi diagram of $n$ point sites with respect to~$d_U$
is~$O(n)$, independent of the choice of
$U$~\cite{a-vdsfg-91,f-vddt-97}. 

The state of knowledge is less satisfactory in three dimensions where
we know near-quadratic complexity bounds for the Voronoi diagram with respect to $d_U$ if $U$
is a constant complexity polytope~\cite{ComplexityOfSubstructures},
and that by making $U$ sufficiently complicated one can have four point
sites define arbitrarily many Voronoi vertices (that is, isolated
points at equal distance from these four sites)~\cite{iklm-95}.  If we grow
equal-radii balls (for~$d_U$) centered in each of the~$p_i$
simultaneously, every hole in the union of these balls must contain a
Voronoi vertex. 
Theorem~\ref{thrm:cubic voids} therefore implies that one can also
find $\Omega(n^3)$ different quadruples of point sites defining
Voronoi vertices.

\begin{corollary}
  There exist a convex distance function $d_U$ and a set $P$ of $n$
  points in $\R^3$ such that $\Omega(n^3)$ different quadruples of $P$
  define a Voronoi vertex in the Voronoi diagram of $P$ with respect
  to~$d_U$.
\end{corollary}

\paragraph{Notation.} 

For an integer $n$ we let $[n]$ denote the set $\{1,\dots,n\}$. If
$A$ and $B$ are two subsets of $\R^d$ we denote by $A \minksum B$
their \emph{Minkowski sum}
\[ 
A \minksum B = \{\vec{a}+\vec{b} : \vec{a} \in A, \vec{b} \in B\}
\]
and by $\conv(A)$ the \emph{convex hull} of $A$. 


We refer to coordinates in three space by $x,y,z$. 
We refer to the positive $z$-direction as ``upward,'' and to the
positive $y$-direction as ``forward.''  For an object
$A \subset \R^3$, the ``front'' boundary of~$A$ is the upper
$y$-envelope.  Similarly, the ``back'' boundary is the lower
$y$-envelope.

\paragraph{Remark on figures.} 

The reader should take heed that in several places we give explicit
coordinates for a construction and provide a figure, where the figure
depicts the qualitative geometric features of interest using different
coordinates.  Using the coordinates chosen for convenience of
computation would have resulted in figures where features are too
small to see.

\FloatBarrier

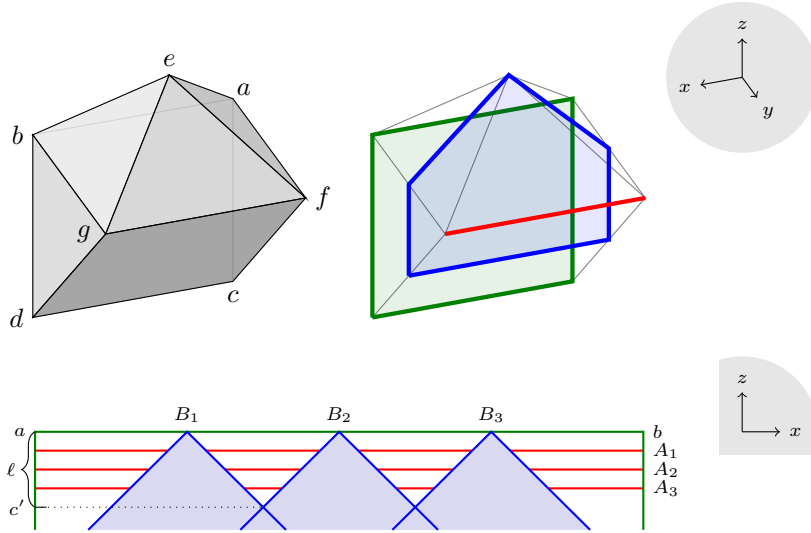
\begin{figure}[ht]
  \centerline{\input{kandslices.tex}}
  \caption{\emph{Top left}: The convex polytope $K$ for the
    construction of Section~\ref{s:ex1}; \emph{Top right}: three
    important $(x,z)$ cross-sections of $K$. The cross section in the
    back (green) is the facet $F$ that touches many holes, the one in the
    center (blue) is used to produce vertical cones, and the one in
    the front (red) is used to produce horizontal edges.
    \emph{Bottom}: The grid formed on the facet $F$ by horizontal
    segments and vertical cones. \label{f:KandSlices}}
\end{figure}

\begin{figure}[htb]
\vspace{-1.5cm}
\centering
\input{unionB.tex}
\hspace{1cm}
\input{unionA.tex}

  \caption{The union of the translates $B_1, B_2, \ldots, B_m$ (left)
    and $A_1, A_2, \ldots, A_m$ (right). The spacing between the
    translates is exaggerated for clarity.\label{fig:union}}
\end{figure}
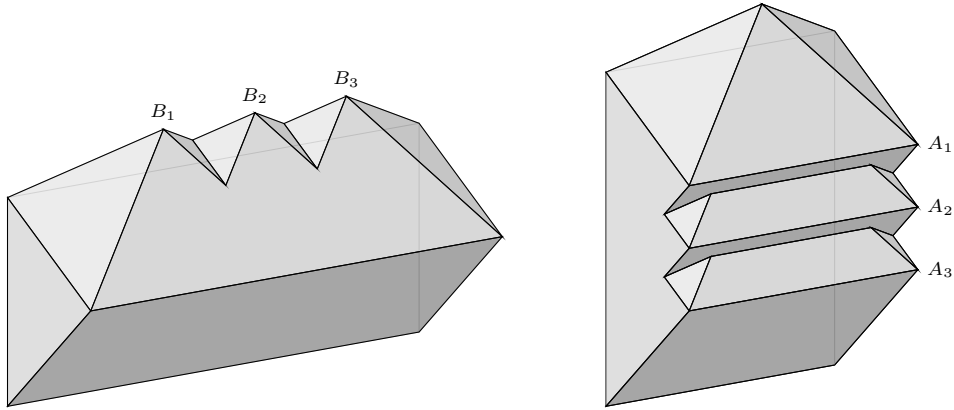

\section{Many holes touching a single facet}
\label{s:ex1}

We first introduce the key idea of our constructions
with a simpler goal: constructing $2m+1$ translates of a convex
polytope with $\Omega(m^2)$ holes in the union, all incident to a
common facet.

Let $K$ be the polytope depicted in Figure~\ref{f:KandSlices},
the convex hull of the following seven points:
\begin{gather*} 
  \vec{a} = (0,0,0),\, \vec{b}= (1,0,0),\, \vec{c}=(0,0,-1),
  \vec{d}=(1,0,-1), \\
 \vec{e}=\left(\half,\half,\half\right),\, \vec{f}=(0,1,0),\, \hbox{
   and } \vec{g}=(1,1,0).
\end{gather*}

We fix an integer $m>0$ and let $C = K$ be the trivial translate of $K$.  Let
$F$ denote the facet of $C$ with vertices $\vec{a},\vec{b},\vec{c},\vec{d}$.  We then
pick $m$ translates of $K$, denoted $B_1, B_2, \ldots, B_m$, whose top-most
vertices (corresponding to vertex $\vec{e}$) are placed regularly
along the edge $\vec{ab}$ of $F$.  The top part of the
intersection~$B_j \cap F$ is a triangular region, shown in blue in
Figure~\ref{f:KandSlices}~(bottom).  
The union $\bigcup_{j=1}^m B_j$ bounds $m-1$ 
regions below the edge $\vec{ab}$ of $F$. Let $\ell$ denote the 
height of one such region; that is the distance between $B_1 \cap B_2$ and the edge $\vec{ab}$ of $F$. 
Let $\vec{c}'$ denote the point on the segment 
$\vec{ac}$ of $F$ at distance $\ell$ from $\vec{a}$ (see again
Figure~\ref{f:KandSlices}~(bottom)). We next pick $m$ translates of
$K$, denoted $A_1, A_2, \ldots, A_m$, whose vertices corresponding to $\vec{f}$ are
placed regularly along the segment $ac'$.

The intersections $F \cap B_j$ and $F \cap A_i$, for $i,j \in [m]$,
form a grid in $F$ with $m(m-1)$ holes on~$F$.  Since two
consecutive~$A_{i}$ leave only a narrow tunnel incident to~$F$, and each
$B_{j}$ entirely cuts this tunnel, each of the holes on~$F$ is
indeed on the boundary of a distinct hole in the union of all
translates in~$\mb{R}^3$.
\begin{claim}
  \label{prop}
  The union $U = \bigcup_{i=1}^m A_i \cup
  \bigcup_{j=1}^m B_j \cup C$ has $\Omega(m^2)$ holes, all
  touching the facet $F$ of $C$.
\end{claim}
Since this is only a warm-up example, we do not include a formal proof
of the claim.

\FloatBarrier

\section{The construction of~$K_m$}
\label{s:ex2}

The construction of Section~\ref{s:ex1} uses three features of $K$: a
portion of a cone with apex $\vec{e}$, a portion of a prism with edge
$\vec{fg}$, and a facet $F$. We combined the $A_i$'s and the
$B_j$'s in a grid-like structure that created $\Theta(m^2)$ local
minima in the $y$-direction on the front boundary of the union $\bigcup A_i
\cup \bigcup B_j$.  Each of these minima is the bottom of a ``pit'',
and $C$ acts as a ``lid'' to turn each pit into a separate hole.  The
construction we use to prove Theorem~\ref{thrm:cubic voids} is based
on a similar principle, but before giving this construction, 
we first fix a value $m > 0$, and then build a convex polytope~$K_{m}$
depending on~$m$ and a family of $3m$ translates with $\Theta(m^3)$ holes.  
The construction of~$K_m$ consists of two parts,
which we refer to as \Front~and \Back.

\paragraph{The auxiliary paths.} 

To construct our polytope $K_m$ we use two auxiliary polygonal
paths~$\eta$ and $\gamma$. They both start at the origin $\vec{0}$.
The path $\eta$ has $m$ edges, lies in the $(y,z)$-plane, and is
convex\footnote{We say a path~$\pi$ is convex in a \emph{direction}
  $u$ if the orthogonal projection of $\pi$ onto $u^\bot$ is injective and the set $\pi + \R^+ u$ is convex.} in both
directions~$(0, -1, 0)$ and $(0, 0, -1)$.  The path $\gamma$ has $m+2$
edges, lies in the $(x,y)$-plane, and is convex in
direction~$(0,-1,0)$.  We denote the $j$th vertex of $\eta$ by
$\vec{w}_{j,0}$ and the $k$th vertex of $\gamma$ by $\vec{w}_{0,k}$,
see the top left of Figure~\ref{f:front}.

\paragraph{The front part.}

We define a ``grid'' of $(m{+}1) \times (m{+}3)$ points, which are the
vertices of the polyhedral surface $\eta \minksum \gamma$
(see the top right of Figure~\ref{f:front}), by putting
\[ 
\vec{w}_{j,k} = \vec{w}_{j,0}+\vec{w}_{0,k} \quad \text{ for }
j\in \{0,1, \ldots, m\} \text{ and } k \in \{0,1, \ldots, m{+}2\}. 
\] 

We then add a point~$v_{j,k}$ on the edge $w_{j,k}w_{j,k+1}$ as follows 
(see the bottom left of Figure~\ref{f:front}):
\[  \vec{v}_{j,k} = \big(1-\tfrac{j}{m+1}\big)\vec{w}_{j,k} +\tfrac{j}{m+1} \vec{w}_{j,k+1}
\quad \text{ for } j\in \{0,1, \ldots, m\} \text{ and } k \in \{0,1, \ldots, m{+}1\}. 
\]
For $j \in \{0,\dots,m\}$ we define two polygonal paths
\begin{align*}
\tpaths_{j} & = w_{j,0} + \gamma = w_{j,0} w_{j,1} w_{j,2} \dots w_{j,m+1} w_{j,m+2} \text{ and}  \\
\hpaths_{j} & = w_{j,0} v_{j,0} v_{j, 1} v_{j,2} v_{j,3}\dots v_{j,m} v_{j, m+1} w_{j,m+2}.
\end{align*}
Note that~$\hpaths_0 = \tpaths_0 = \gamma$, 
that the paths~$\tpaths_{j}$ are simply translates of~$\gamma$, and
that the path~$\hpaths_{j}$ lies entirely in the convex region~$\tpaths_{j} + (0,\R^-,0)$ 
and the vertices of~$\hpaths_{j}$ lie on~$\tpaths_{j}$; 
see Figure~\ref{fig:gammas}.

The front part~$\Front$ of~$K_m$ is the convex hull of the
paths~$\hpaths_{j}$ (see the bottom right of Figure~\ref{f:front}):
\[ 
\Front = \conv\big( \hpaths_0 \cup \hpaths_1 \cup \dots \cup \hpaths_m
\big).
\]
Observe that the paths $\eta$ and $\gamma$ can be chosen so that, for
each $j \in [m]$, the path $\hpaths_j$ lies entirely on the front
boundary as well as on the upper boundary of~$\Front$.

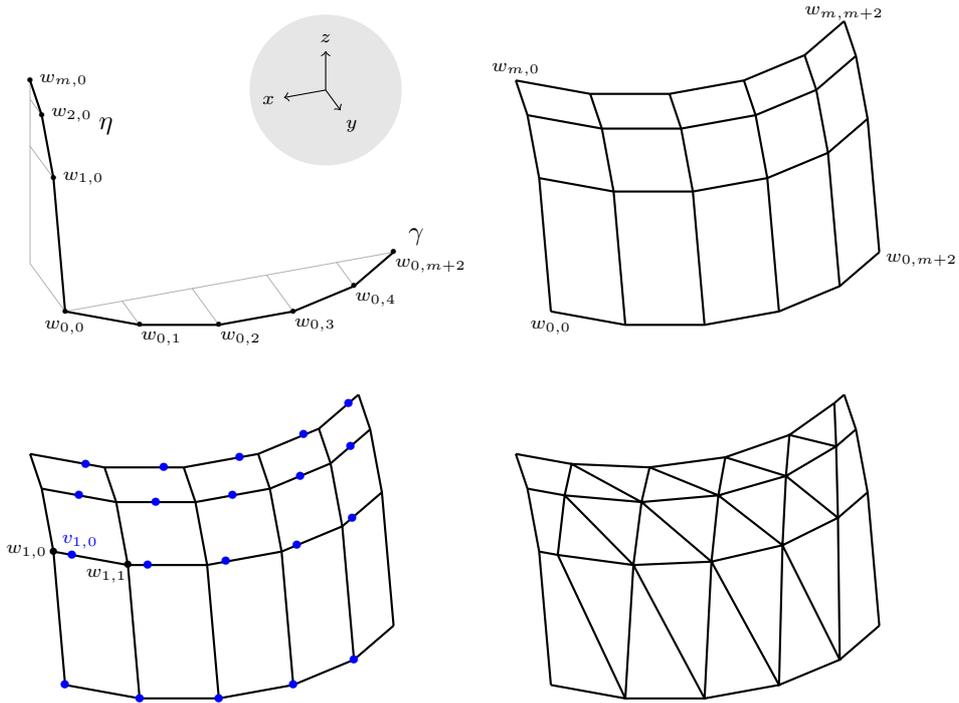
\begin{figure}[htp]
  \centerline{\input{front.tex}}
  \caption{Design of $\Front$.\label{f:front}
  \emph{Top left}: The auxiliary paths $\eta$ and $\gamma$.  
  \emph{Top right}: The surface $\eta + \gamma$ defining the ``grid'' of vertices $w_{j,k}$. 
  \emph{Bottom left}: The points $v_{j,k}$. 
  \emph{Bottom right}: The front boundary of $\Front$.
  }
\end{figure}

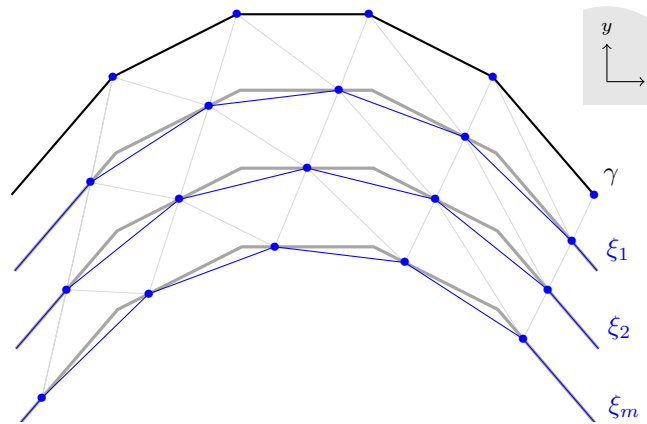
\begin{figure}[htp]
  \centerline{\input{gammas.tex}}
  \caption{View of the paths~$\hpaths_j$ from above (with the paths $\tpaths_j$ represented in grey).\label{fig:gammas}}
\end{figure}

Let $E_{j,k}$ denote the edge~$w_{j,k}w_{j,k+1}$ and let~$E_k = E_{0,k}$.  Consider the point~$v_{j,k}$ on $E_{j,k}$. 
Since $\hpaths_j$ lies entirely on the upper boundary of~$\Front$, no part of~$\Front$ appears above $E_{j,k}$,
and since $v_{j,k}$ is the only point where the segment~$E_{j,k}$ intersects $\Front$ for $j,k \in [m]$, 
the vertical plane containing~$E_{j,k}$ 
intersects $\Front$ in a downward extending cone with apex~$v_{j,k}$;
see Figure~\ref{fig:vert}.  Note this fails for $k=0,m{+}1$. 
\begin{figure}[htp]
  \centerline{\input{vert_cones_Km.tex}}
  \caption{Intersection of vertical panels through $E_{1,k}$ with~$\Front$.\label{fig:vert}}
\end{figure}
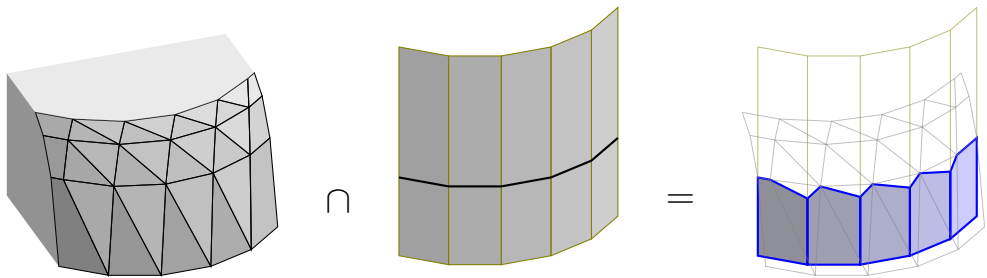

\FloatBarrier

\paragraph{The back part.}

We now define two vectors, $\ut = w_{m,m+2} - (0,t,0)$, where $t$ is
some positive real number, and $\uz$, the orthogonal projection of
$\ut$ on the $(x,y)$-plane. We define $\Back$ as the Minkowski sum of
the segment $\uz\ut$ and $- \gamma$, the reflection of~$\gamma$ with
respect to the origin (see Figure~\ref{f:back}). The value of~$t$ is
adjusted so that $\Front$ and $\Back$ have disjoint convex hulls.
\begin{figure}[htb]
  \centerline{\input{back.tex}}
  \caption{Design of $\Back$.\label{f:back}}
\end{figure}
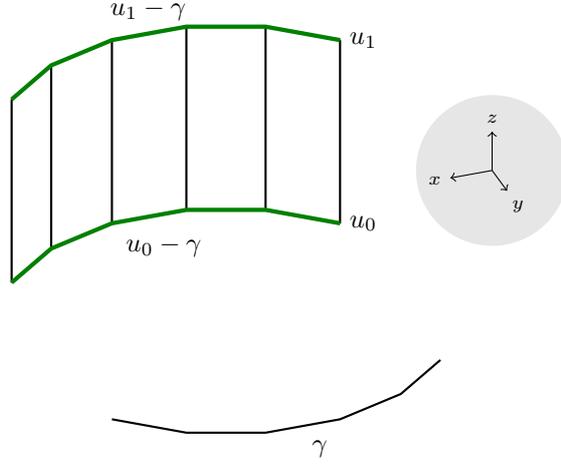

By construction, $\Back$ consists of $m+2$ rectangles orthogonal to the
$(x,y)$-plane, namely the rectangles $u_{0}u_{1} - E_k$, for $k \in
\{0,\dots,m+1\}$.  The top and bottom edges of each rectangle are $\ut
- E_{k}$ and $\uz - E_{k}$ respectively.

\begin{figure}[H]
 \centerline{\input{kex2.tex}}
  \caption{\emph{Left}: The polytope $K_m$. \emph{Right}: $K_m$ as viewed from behind. \label{f:Kex2}}  
\end{figure}
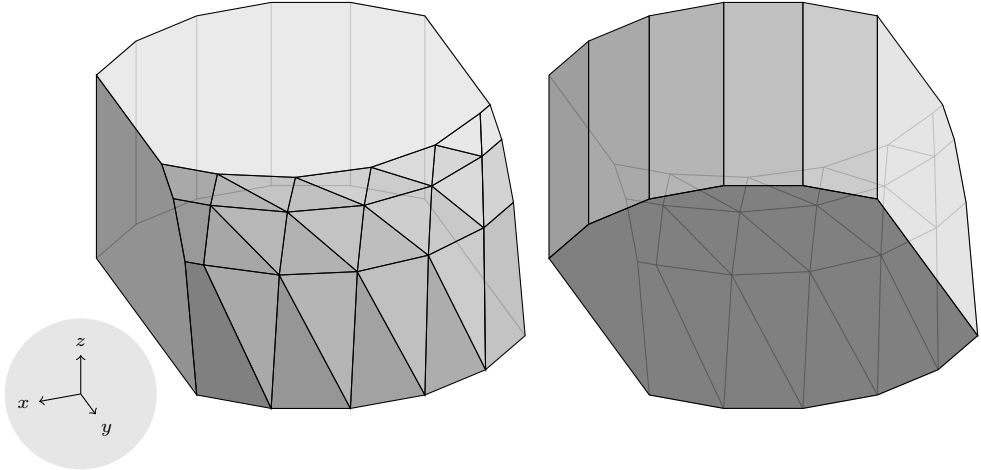

\begin{figure}[p]
   \centerline{\input{trans.tex}}
  \caption{
  \emph{Top}: A translate~$A_{i}$.
  \emph{Middle left}: A translate~$B_{j}$. \emph{Middle right}: The intersection of the vertical panels though $E_1,\dots,E_m$ with the translate $B_j$.
  \emph{Bottom left}: A translate~$C_{k}$ and the facet $R_k$. \emph{Bottom right}: $C_{k}$ and $R_k$ as viewed from behind.
  \label{f:trans}}  
\end{figure}
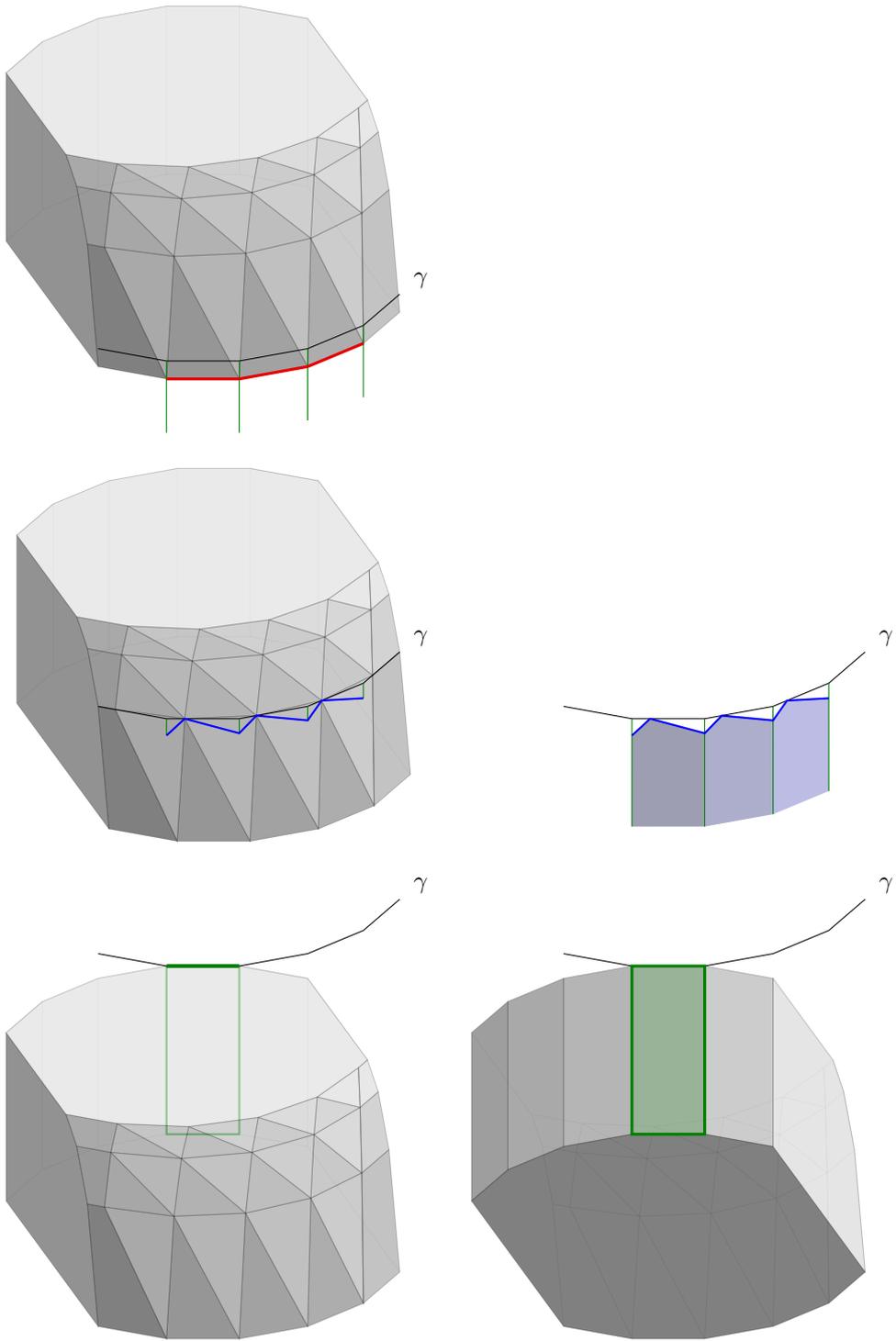

\paragraph{The polytope $K_m$ and its translates.}

We now let $K_m = \conv(\Front \cup \Back)$, see Figure~\ref{f:Kex2},
and define three families of its translates.  First, for $k \in [m]$,
we define a translate $C_{k}$ such that the edge $\ut - E_{k}$
of~$C_{k}$ coincides with the edge~$E_{k}$ of~$K_{m}$.  Formally, we
have
\[ 
C_k = K_{m} + \vec{c}_k, \quad \hbox{where} \quad \vec{c}_k =
\vec{w}_{0,k} - (\ut - \vec{w}_{0,k+1}).
\]
See Figure~\ref{f:trans}~(bottom).  The vertical facet formed by $\ut -
E_{k}$ and $\uz - E_{k}$ of~$C_{k}$ will be incident to a quadratic
number of holes, playing the same role as the facet~$F$ in the
previous section.  Let us denote this facet as~$R_{k}$.

Next, for $j \in [m]$, we define a translate~$B_{j}$
of~$K_{m}$ as follows:
\[ 
B_j = K_{m} + \vec{b}_j, \quad \hbox{where} \quad \vec{b}_j  = -
\vec{w}_{j,0}.
\]
In other words, the path $\hpaths_{j}$ of~$B_{j}$ lies in the
$(x,y)$-plane and the vertex~$v_{j,k}$ of~$B_{j}$ lies on~$E_{k}$.
See Figure~\ref{f:trans}~(middle).

Consider now the edge~$E_{k}$, for $k \in [m]$. For each $j \in [m]$,
the edge $E_{k}$ contains the vertex $v_{j,k}$ of~$B_{j}$.  By the
argument above, the intersection of $B_{j}$ with the facet~$R_{k}$ is
a vertical cone with apex~$(1-\nicefrac{j}{m+1})w_{0,k} +
(\nicefrac{j}{m+1})w_{0,k+1}$. These apices are regularly spaced
along~$E_{k}$, and we obtain a configuration similar to
Figure~\ref{f:KandSlices}~(bottom).  In other words, the union
$\bigcup_{j\in [m]}B_{j}$ bounds $m-1$ triangular regions below the
edge~$E_{k}$ on the facet~$R_{k}$ of~$C_{k}$.

We now pick a suffiently small number $\eps > 0$ and define our final
family of translates. For $i \in [m]$, let
\[  
A_i = K + \vec{a}_i, \quad \text{where } 
\vec{a}_i = \big(0,0,-\tfrac{i}{m} \eps \big).
\]
The translates $A_{i}$ are defined by translating $K$ vertically such that 
for every $i\in [m]$ and $k\in [m]$, the edge~$E_{k}$ of~$A_{i}$ appears
as a horizontal edge on~$R_{k}$, cutting each of the $m-1$ triangular
regions.
See Figure~\ref{f:trans}~(top).

\paragraph{The family~$\F$.}
We now finally set 
\[
\F = \{A_1,\dots,A_m,\, B_1,\dots,B_m,\, C_1,\dots,C_m\}.
\]
The union of the family~$\F$ has at least $m^2(m-1)$ holes: For $k \in
[m]$, we consider the facet~$R_{k}$ of~$C_{k}$.  On this facet, the
union of the $B_{j}$ and~$A_{i}$ forms a grid with~$m(m-1)$ holes.
As in Section~\ref{s:ex1}, we argue that each of these holes is
incident to a distinct component of the complement of the union of all
the translates.

We will not give a formal proof of this fact, since we will give an
even stronger construction in Section~\ref{sec:universal}, and we will
include a formal, algebraic argument for the correctness of that
construction.

\paragraph{Explicit coordinates for~$K_m$.}

The reader not satisfied with the qualitative description of~$\F$ may
enjoy verifying that the following coordinates satisfy the properties
we needed for our construction:
\begin{align*}
  \vec{w}_{j,0} & = (0,-3j,1-2^{-j}) & \text{ for } & 0 \leq j \leq m, \\ 
  \vec{w}_{0,k} & = (\cos\theta_k-1, \sin\theta_k, 0) & \text{ for } &
  \theta_k = \tfrac{\pi}{3}\left(\tfrac{k-1}{m} +1\right), 1 \leq k \leq m+1, \\
  \vec{w}_{0,m+2} & =  (-2,0,0), \\ 
  \ut &= (-2,-3m-3,1-2^{-m}), \\
  \uz &= (-2,-3m-3, 0).
\end{align*}

\FloatBarrier

\section{Constructing a universal convex body}
\label{sec:universal}

The family of translates constructed in Section~\ref{s:ex2} uses a
convex polytope~$K_{m}$ that depends on~$m$.  In this section we construct
a single convex body~$K$ that allows the formation of families of $n$
translates of~$K$, for arbitrarily large~$n$, with a cubic number of
holes. (Note that the \emph{position} of the translates in the family
will depend on~$n$.)

\paragraph{The convex body~$K$.}

The finite polygonal paths $\gamma$ and $\eta$ are replaced by
infinite polygonal paths.  We must also redefine the vertices
$v_{j,k}$, but the rest of the construction remains largely the same
as before.  Figure \ref{fig:many-large-grids} illustrates the
construction.
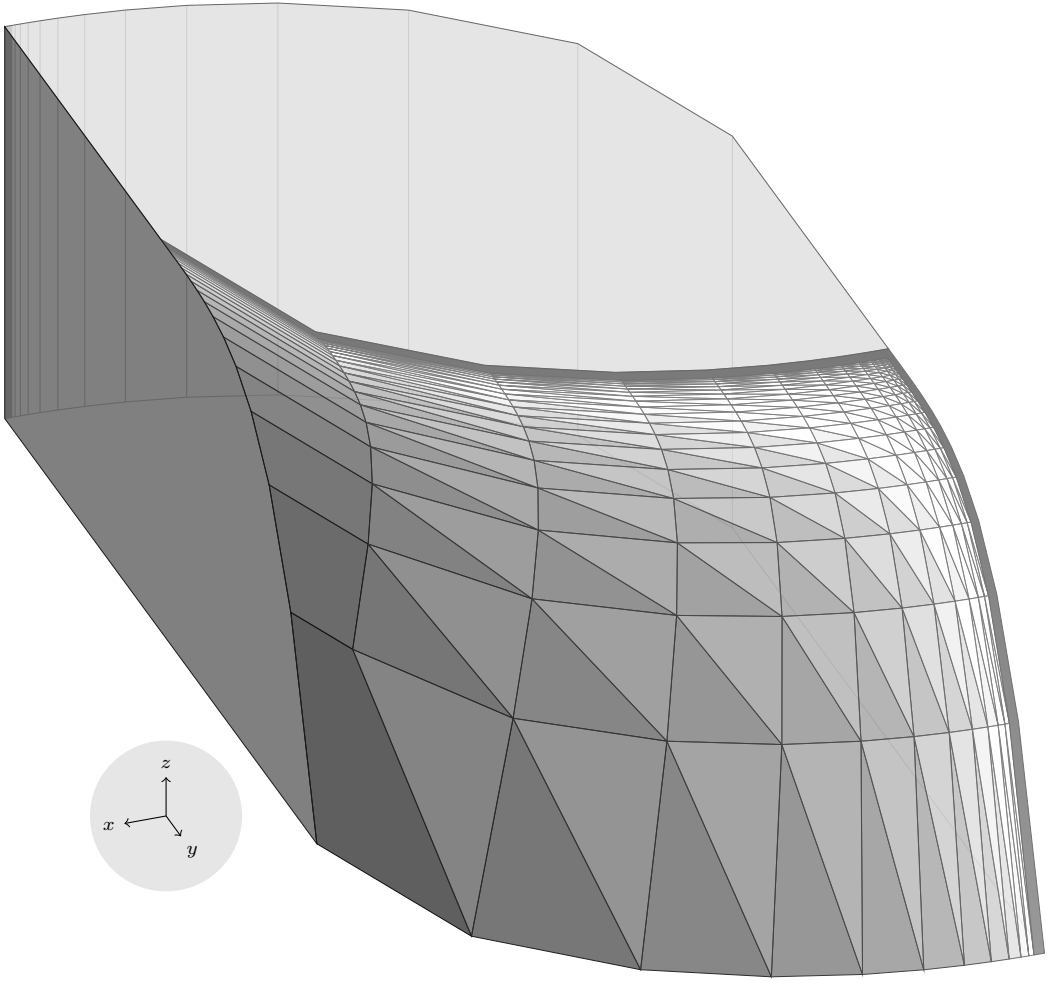
\begin{figure}[t]
  \def\yesyes{\yes}
  \ifx\showuniversal\yesyes\centerline{\input{universal.tex}}
  \else\centerline{OMITTED for speed}\fi
  \caption{A convex body $K$, not depending on $n$, with translates forming
    $\Theta(n^3)$ holes.}
  \label{fig:many-large-grids}
\end{figure}

Using $\zeta_s = \sum_{t=1}^\infty t^{-s}$ for $s \in \{2,3\}$, we
define vertices as follows:
\begin{align*}
  \vec{w}_{0,0} & = \vec{0}, \\
  \vec{w}_{j,0} & = \vec{w}_{j-1,0} + \left(0, -j^{-2}, j^{-3} \right) 
  & \qquad & \text{for $j \in \{1,2,\dots\}$},
  \\
  \vec{w}_{\infty,0} & = \lim_{j \to \infty} \vec{w}_{j,0}  
  = \left(0, -\zeta_2, \zeta_3 \right), \\
  \vec{w}_{0,k} & = \vec{w}_{0,k-1} + \left(k^{-2}, k^{-3}, 0  \right) 
  & & \text{for $k \in \{1,2,\dots\}$}, \\
  \vec{w}_{0,\infty} & = \lim_{k \to \infty} \vec{w}_{0,k}  
  = \left(\zeta_2, \zeta_3, 0 \right), \\
  \vec{w_{j,k}} & = w_{j,0} + w_{0,k} 
  & & \text{for $j,k \in \{1,2,\dots,\infty\}$}, \\
  \vec{v}_{j,k} & = \tfrac{1}{(j+1)^3} \vec{w}_{j,k} + 
  \tfrac{(j+1)^3-1}{(j+1)^3} \vec{w}_{j,k+1} 
  & & \text{for $j,k \in \{0,1,\dots\}$}, \\
  \ut & = (\zeta_2,-2,\zeta_3), \\
  \uz & = (\zeta_2,-2,0).
\end{align*}
We define the convex path~$\gamma$ as $w_{0,0}w_{0,1}w_{0,2}\dots
w_{0,\infty}$, and the convex path~$\eta$ as $w_{0,0} w_{1,0} w_{2,0}
\dots w_{\infty, 0}$.  Again we let $E_{k}$ denote the edge
$w_{0,k}w_{0,k+1}$, for $k \in \{0,1,\dots\}$.  The vertex $v_{j,k}$
lies on the edge $w_{j,k}w_{j,k+1} = E_{k} + w_{j,0}$.  For $j \in
\{1,2,\dots\}$, we set $\tpaths_{j} = w_{j,0} + \gamma$, and define the
convex path $\hpaths_{j}$ as $w_{j,0}v_{j,0}v_{j,1}v_{j,2}v_{j,3}\dots
w_{j,\infty}$; note that $\hpaths_0 = \gamma$.  The path $\hpaths_\infty$ is equal to $\tpaths_{\infty}
= \gamma + w_{\infty, 0}$; note that $\lim_{j \to \infty} v_{j,k} =
w_{\infty, k+1}$.

The front part of~$K$ is the convex hull of the paths~$\hpaths_{j}$:
\[
\Front = \conv(\hpaths_0 \cup \hpaths_1 \cup \hpaths_2 \cup \dots \cup
\hpaths_{\infty}).
\]
The back part $\Back$ of~$K$ is the Minkowski sum of~$\ut\uz$ and
$-\gamma$. By construction, it is the union of the rectangles $\ut\uz
- E_{k}$, for $k \in \{0,1,2,\dots\}$.  The top and bottom edges of
each rectangle are $\ut - E_{k}$ and $\uz - E_{k}$.  

Finally, we define $K = \conv(\Front \cup \Back)$, concluding the
description of our convex body~$K$.  This body has the following
property.
\begin{lemma}
  The polygonal path~$\hpaths_{j}$ lies entirely on the front boundary of~$K$.
\end{lemma}
\begin{proof}
  We will show that any point in $\hpaths_{j}$ lies outside the convex
  hull of the regions $\hpaths_{j'} +(0,\R^-,0)$, for $j' \neq j$.  
  Fix $j \in \mb{N}$, and let $\KK$ be the convex hull of the regions $\tpaths_{j'} +(0,\R^-,0)$, for $j' \neq j$.
  Since the path $\hpaths_{j'}$ lies in the convex region $\tpaths_{j'} +(0,\R^-,0)$, it suffices to show that $\hpaths_j$ lies outside $\KK$.
  Since the $\tpaths_{j}$ are just translates of~$\gamma$, the
  body~$\KK$ is easy to describe.  In particular, between the
  $(x,y)$-parallel planes containing~$\tpaths_{j-1}$ and~$\tpaths_{j+1}$, its
  front boundary is formed by rectangles that are the convex hull
  of~$E_{k} + w_{j-1,0}$ and~$E_{k} + w_{j+1,0}$.  Let $\Pi$ be the
  $(x,y)$-parallel plane containing~$\hpaths_{j}$ and~$\tpaths_{j}$.  The
  front boundary of~$\KK \cap \Pi$ is again a translate of~$\gamma$.
  We first compute this translate of~$\gamma$, by finding the
  intersection point~$p$ of the segment $w_{j-1,0}w_{j+1,0}$
  with~$\Pi$.

  Let $p' = p - w_{j-1, 0}$.  This makes $p'$ the point at
  height~$1/j^{3}$ on the line through the origin and the point
  $w_{j+1,0} - w_{j-1,0}$.  Since
  \[
  w_{j+1,0} - w_{j-1,0} = \bigg(0,
  -\frac{1}{j^{2}}-\frac{1}{(j+1)^{2}}, 
  \frac{1}{j^{3}}+\frac{1}{(j+1)^{3}} \bigg) =
  \bigg(0,  \frac{-j^2 -(j+1)^2}{j^2(j+1)^2}, 
  \frac{j^3 +(j+1)^3}{j^3(j+1)^3} \bigg),
  \] 
  this gives
  \[
  p' = \bigg(0,\, \frac{-(j+1)(j^2+(j+1)^2)}{j^2(j^3 + (j+1)^3)},\, 
  \frac{1}{j^3} \bigg).
  \]
  Since $p = p' + w_{j-1,0}$ and 
  $w_{j,0} - w_{j-1,0} = (0, -j^{-2}, j^{-3})$, we have 
  \[
  p = \vec{w}_{j,0} - \Big(0,\, \frac{1}{j^3 +(j+1)^3},\, 0 \Big). 
  \]
  
  We have just computed the front boundary $\gamma + p$ of
  $\KK \cap \Pi$, and we want to show that the path~$\hpaths_j$ 
  lies farther in front of this boundary. Since $\gamma + p$ and $\hpaths_j$
  are convex paths, it suffices to show that any vertex of $\gamma +
  p$ lies behind~$\hpaths_{j}$.  These vertices are the
  points~$w_{0,k} + p$, for $k \in \{0, 1, 2, \dots\}$.
  That is, we will show~$w_{0,k} +p \in \hpaths_j +(0,\R^-,0)$.  

  We fix some $k \in \{1, 2, \dots\}$ and consider~$w_{0,k} + p$.  The
  line parallel to the $y$-axis through $w_{0,k} + p$ intersects the
  edge $v_{j,k-1}v_{j,k}$ of~$\hpaths_{j}$ in a point~$q$.  We need to
  show that $q^{y} \geq (w_{0,k} + p)^{y}$ (here and in the following,
  we use superscripts~$x,y,z$ to denote the coordinates of a point).

  It will be convenient to translate our coordinate system such that
  $w_{j,k}$ is the origin.  This means that $\Pi$ is the plane~$z =
  0$.  Letting~$J = (j+1)^{3}$, we have:
  \begin{align*}
    w_{0,k} + p - w_{j,k} & = w_{0, k} + p - (w_{j,0} + w_{0,k}) 
    = \Big(0, \frac{-1}{j^3
      + J}, 0 \Big), \\
    v_{j,k} - w_{j,k} & = \tfrac{1-J}{J} w_{j,k} + 
    \tfrac{J-1}{J} w_{j,k-1} 
    = \tfrac{J-1}{J}\big[w_{j,k+1} - w_{j,k}\big], \\
    v_{j,k-1} - w_{j,k} & = \tfrac{1}{J} w_{j,k-1} 
    - \tfrac{1}{J} w_{j,k} 
    = \tfrac{1}{J} \big[w_{j,k-1} - w_{j, k}\big].
  \end{align*}

  Now parameterize the segment~$v_{j,k-1}v_{j,k}$ as~$E(s)$,
  for $0 \leq s \leq 1$, using our new coordinate system:
  \begin{align*}
    E(s) &= (1-s)\frac{J-1}{J}\big[w_{j,k+1} - w_{j,k}\big] 
    + \frac{s}{J} \big[w_{j,k-1} - w_{j, k}\big].
  \end{align*}
  The $x$- and $y$-coordinates of the point~$E(s)$ are
  \begin{align*}
    E(s)^{x} & = (1-s)\frac{J-1}{J}\frac{1}{(k+1)^{2}} 
    - \frac{s}{J}\frac{1}{k^{2}}, \\
    E(s)^{y} & = (1-s)\frac{J-1}{J} \frac{1}{(k+1)^{3}} 
    - \frac{s}{J} \frac{1}{k^{3}}.
  \end{align*}

  We have $q - w_{j,k} = E(t)$, for the $t \in [0,1]$ where $E(t)^{x}
  = 0$.  This condition is equivalent to $(1-t)(J-1)k^{2} = t(k+1)^{2}$,
  and therefore
  \begin{align*}
    t & = \frac{(J-1)k^2}{ (k+1)^2 +(J-1)k^2}, 
    \\
    E(t)^{y} & = \frac{1-J}{J}\frac{1}{k(k+1)((k+1)^2 +(J-1)k^2)}
    \geq \frac{1-J}{J}\frac{1}{2(J + 3)},
  \end{align*}
  where we used $k(k+1) \geq 2$, $k^{2} \geq 1$, $(k+1)^{2} -
  k^{2} \geq 3$, and the fact that $1-J \leq 0$.
  We further have
  \begin{align*}
    (w_{0,k} + p - w_{j,k})^{y} & = \frac{-1}{j^{3} + J}
    < \frac{-1}{2J} < \frac{1-J}{2J(J + 3)} 
    \leq E(t)^{y} = (q - w_{j,k})^{y},
  \end{align*}
  and therefore $(w_{0,k}+p)^{y} < q^{y}$.
  Thus, $w_{0,k}+p \in \hpaths_j + (0,\R^-,0)$, so $\hpaths_j$ is in
  front of $\KK$ and as such is on the front boundary of $K$. 
\end{proof}

\paragraph{The translates.}

We now pick a number $m\in \N$ and construct a family~$\F$ of
$3m$~translates of~$K$ such that their union will have a cubic number
of holes.  This construction is identical to the construction in
Section~\ref{s:ex2}:

First, for $k \in [m]$, we define a translate $C_{k}$ such that the
edge $\ut - E_{k}$ of~$C_{k}$ coincides with the edge~$E_{k}$
of~$K$, that is
\[ 
C_k = K + \vec{c}_k, \quad \hbox{where} \quad \vec{c}_k =
\vec{w}_{0,k} - (\ut - \vec{w}_{0,k+1}).
\]
Again we denote the vertical facet formed by $\ut - E_{k}$ and $\uz -
E_{k}$ of~$C_{k}$ as~$R_{k}$.

Next, for $j \in [m]$, we define a translate~$B_{j}$
of~$K$ as follows:
\[ 
B_j = K + \vec{b}_j, \quad \hbox{where} \quad \vec{b}_j  = -
\vec{w}_{j,0}.
\]
In other words, the path $\hpaths_{j}$ of~$B_{j}$ lies in the
$(x,y)$-plane, the vertex~$v_{j,k}$ of~$B_{j}$ lies on~$E_{k}$.

For the third group of translates we need to determine a sufficiently
small~$\eps > 0$.  First, observe that the points $w_{j,k}$, for $j,k
\in [m]$, do not lie on~$K$.  Let $\eps_{1}> 0$ be smaller than the
distance of $w_{j,k}$ to~$K$, for all $j,k \in [m]$.  Second, consider
the $m^{2}$ points
$ v_{j,k} +b_j$ for $j, k \in [m]$ 
on the path~$\gamma$. Let $\eps_{2}$ be the
shortest distance between any two of these points.

Consider now the segment $w_{j,k}w_{j,k+1}$, for some $j, k \in [m]$.
It touches $K$ in the point~$v_{j,k}$, the rest of the segment lies
entirely outside~$K$.  This implies that there is an $\eps_{j,k} > 0$
such that any line parallel to $w_{j,k}w_{j,k+1}$ at distance less
than $\eps_{j,k}$ intersects~$K$ only within a neighborhood
of~$v_{j,k}$ of radius~$\eps_2/3$.

We choose $\eps < \eps_1$ and $\eps < \eps_{j,k}$, for all $j,k \in
[m]$.  With this choice of~$\eps > 0$, we can finally define, for $i
\in [m]$:
\[  
A_i = K + \vec{a}_i, \quad \text{where } 
\vec{a}_i = \big(0,0,-\tfrac{i}{m} \eps \big).
\]

Our family~$\F$ is
\[
\F = \{A_1,\dots,A_m,\, B_1,\dots,B_m,\, C_1,\dots,C_m\}.
\]

\paragraph{The nerve.}
To every family $\mathcal{X}=\{X_1,X_2, \ldots, X_n\}$ of $n$ sets is
associated a collection of subfamilies $\nerv(\mathcal{X})$, called
the \emph{nerve} of $\mathcal{X}$, defined as follows:
\[
\nerv(\mathcal{X}) = \left\{\mathcal{Y} \subseteq \mathcal{X}: \bigcap_{X \in \mathcal{Y}}X \neq
\emptyset \right\}.
\] 
In a sense, the nerve is a natural generalization of the 
\emph{intersection graph}. 
In Section~\ref{sec:nerve}, we will count the number of holes in the union of $\F$ by computing the rank of certain matrices defined in terms of its nerve. 
We now give an explicit description of the nerve of the family~$\F$. 
Consider the following subfamilies of $\F$:
\begin{align*}
  \Delta_1 & =  \{A_1,\dots,A_m,B_1,\dots,B_m\}, \\
  \Delta_2 & =  \{B_1,\dots,B_m,C_1,\dots,C_m\}, \\
  \Delta_{i,k} & =  \{A_{i}, C_{k}, C_{k+1}\} &  \text{for } & (i,k) \in
        [m]\times [m-1],\\ 
        \Delta_{i,j,k} & =  \{A_{i}, B_{j}, C_{k}\} &  \text{for } &
        (i,j,k) \in [m]^3. 
\end{align*} 

\begin{lemma}\label{l:nerve}
  The set of inclusion-maximal subfamilies in $\nerv(\F)$ is
  \[
  \mathcal{M} = \{\Delta_1, \Delta_2\} \cup \{\Delta_{i,k} : (i,k) \in
  [m]\times[m-1]\} \cup \{\Delta_{i,j,k}: (i,j,k) \in [m]^3\}.
  \]
\end{lemma}
\begin{proof}
  To check that the subfamilies in $\mathcal{M}$ are in $\nerv(\F)$, we find a
  point in the intersection of each of them; see Appendix~\ref{apx:universal} for details.  Now,
  let $\sigma$ be a maximal subfamily in $\nerv(\F)$. If $\sigma$ does
  not contain any $C_k$ then $\sigma \subseteq \Delta_1$ and by
  maximality $\sigma = \Delta_1$. Similarly, if $\sigma$ does not
  contain any $A_i$ then $\sigma \subseteq \Delta_2$ and by maximality
  $\sigma = \Delta_2$.

  We can therefore assume that $A_i, C_k \in \sigma$.  By definition
  of $K$ and $\F$ we have $A_i \cap C_k = E_k +\vec{a}_i$.  
  Since $A_i \cap C_k$ and $A_{i'} \cap C_{k}$ are parallel
  segments for $i' \neq i$, $\sigma$ cannot contain~$A_{i'}$.

  Assume now that $\sigma$ contains no $B_j$. The segments $E_k$ and $E_{k'}$ 
  intersect if and only if~$k$ and~$k'$ differ by one. It follows that
  $\sigma$ is either $\Delta_{i,k}$ or $\Delta_{i,k-1}$.

  In the final case, $\sigma$ contains some~$B_j$, for $j \in [m]$.
  The segment $E_{k} + a_{i} - b_{j}$ is parallel to
  $w_{j,k}w_{j,k+1}$ at distance at most $\eps < \eps_{j,k}$, and so
  it intersects $K$ only in a neighborhood of $v_{j,k}$ of radius at
  most~$\eps_2/3$.  It follows that $E_{k}+ a_{i}$ intersects~$B_{j} =
  K + b_{j}$ only in a neighborhood of the same radius around the
  point
  $ v_{j,k} +b_j$. 
  But, the shortest distance between these points is~$\eps_2$, and so
  these neighborhoods are disjoint.  It follows that $\sigma$
  contains no other $B_{j'}$, for $j'\neq j$.

  Since the point~$w_{0,k} + a_{i}$ lies at distance at most~$\eps$
  from $w_{0,k}$, but the point $w_{j,k} = w_{0,k} - b_{j}$ has
  distance larger than~$\eps_1 > \eps$ from~$K$, we have
  $A_i \cap C_{k-1} \cap C_k = w_{0,k} + a_{i} \not\in B_{j}$, and so $C_{k-1} \not\in \sigma$.  For
  the same reason~$C_{k+1} \not\in \sigma$.
  It follows that $\sigma = \{A_{i}, C_{k}, B_{j}\} = \Delta_{i,j,k}$.
\end{proof}

\FloatBarrier

\section{Counting holes in the union via the nerve}
\label{sec:nerve}

In this section we develop a technique for counting the number of
holes in a union of convex objects.  We will demonstrate this technique by giving a new
proof of Kovalev's upper bound~\cite{k-svm-88}, and then we will use it to prove Theorem~\ref{thrm:cubic voids}.  
We start by recalling some standard topological machinery.

\paragraph{Simplicial complex.}

A \emph{simplicial complex} $\SC$ with vertex set $V$ is a set of
subsets of $V$ closed under taking subsets: if $\sigma \in \SC$ and
$\tau \subseteq \sigma$ then $\tau \in \SC$. An element $\sigma \in
\SC$ is called a \emph{simplex}; the \emph{dimension} of a simplex is
its cardinality minus $1$, so singletons are simplices of dimension
$0$, pairs are simplices of dimension $1$, etc... A simplex of
dimension $i$ is called an \emph{$i$-simplex} for short. The
\emph{vertices} of a simplex $\sigma \in \SC$ are the singletons
contained in $\sigma$.  Note that the nerve of a family of convex sets is a simplicial complex.

\paragraph{Homology and Betti numbers.}

Let $\SC$ be a simplicial complex on a \emph{totally ordered} vertex
set $V$.  The $i$th real chain space of $\SC$, denoted $\chain_i(\SC)$, is
the real vector space spanned\footnote{%
In other words, the $i$-dimensional simplices of $\SC$ form a basis of the vector space  $\chain_i(\SC)$, which consists of formal sums of $i$-simplices, each simplex being assigned a real coefficient.} 
by the $i$-simplices of $\SC$.  
For $i \in \N$, 
the \emph{$i$th boundary map} $\partial_i: \chain_i(\SC) \to \chain_{i-1}(\SC)$ is the linear map defined on a basis of $\chain_i(\SC)$ as follows. 
For any $i$-simplex 
$\sigma = \{v_0,v_1, \ldots, v_i\} \in \SC$ with $v_0<v_1< \ldots < v_i$, 
\[ \partial_i(\sigma) = \sum_{j=0}^{i} (-1)^j (\sigma \setminus \{v_j\}).\]
That is, $\partial_i$ maps every $i$-dimensional simplex to an element
of $\chain_{i-1}(\SC)$, namely an alternating sum of its facets. 
Observe that 
$\partial_i \circ \partial_{i+1} = 0$, so that $\im \partial_{i+1} \subseteq \ker \partial_i$. The $i$th \emph{simplicial
  homology group} $H_i(\SC,\R)$ of $\SC$ is defined as the quotient
$\ker \partial_i/ \im \partial_{i+1}$ and the $i$th \emph{Betti
  number} $\beta_i(\SC)$ of $\SC$ is the dimension of $H_i(\SC,\R)$,
hence:
\begin{equation}\label{e:betti}
  \beta_i(S) = \dim \ker \partial_i - \rank  \partial_{i+1}.
\end{equation}

If $X$ is a subset of $\R^d$, one can define, in a similar but more
technical way, the \emph{singular homology groups} of $X$ and its
\emph{Betti numbers}.  We do not recall those definitions (the
interested reader is referred
to~\cite{Hatcher:AlgebraicTopology-2002,Munkres:AlgebraicTopology-1984})
but emphasize two facts that will be useful:
\begin{itemize}
\item[(i)] $\beta_0(X)$ is the number of connected components of $X$, assuming $X$ admits a cell decomposition.
\item[(ii)] If $\mathcal{X}=\{X_1,X_2, \ldots, X_n\}$ is a family of convex
  objects in $\R^d$ and $U = \bigcup_{i=1}^n X_i$ then $H_i(U,\R) \simeq H_i(\nerv(\mathcal{X}))$; 
  as a consequence, $U$ and $\nerv(\mathcal{X})$ have the same Betti numbers. 
  This is the classical \emph{nerve theorem} of Borsuk~\cite{b-iscsc-48}.
\end{itemize}

\paragraph{Counting holes.}

We can now relate the number of holes in the union of a family of
convex objects to one particular Betti number of its nerve.

\begin{lemma}\label{l:hole-betti}
  If $\mathcal{X} = \{X_1,X_2, \ldots, X_n\}$ is a family of compact convex objects in
  $\R^d$ then the number of holes of $U = \bigcup_{i=1}^n X_i$ is
  $\beta_{d-1}(\nerv(\mathcal{X}))+1$.
\end{lemma}
\begin{proof} 
  The number of holes of $U$ is, by definition, the number of
  connected components of $\R^d \setminus U$, which is $\beta_0(\R^d
  \setminus U)$. 
  Assume $d >1$.
  For any compact locally contractible subset $T \subseteq \s^d$,
  Alexander duality gives $\beta_{d-1}(T) = \beta_0(\s^d \setminus T) -1$.  
  Identifying the $d$-sphere with the one-point compactification of $d$-space, $\s^d \simeq \R^d
  \cup\{\infty\}$, we have that 
  \[ \beta_0(\R^d \setminus U) =
  \beta_0(\s^d \setminus U) =
  \beta_{d-1}(U)+1, \]
  and by the nerve theorem, $\beta_{d-1}(U) = \beta_{d-1}(\nerv(\mathcal{X}))$.
\end{proof}

As an illustration let us see how a version of the upper bound of
Kovalev~\cite{k-svm-88} for compact convex objects immediately follows
from Lemma~\ref{l:hole-betti}:

\begin{corollary}\label{c:arb}
  The number of holes in the union of $n$ compact convex
  objects in~$\R^d$ is at most~$\binom{n}d +1$.
\end{corollary}

\begin{proof}
  Let $\mathcal{X}$ be a family of $n$ compact convex objects in $\R^d$. By
  Lemma~\ref{l:hole-betti}, the number of holes of the union of the
  members of $\mathcal{X}$ is $\beta_{d-1}(\nerv(\mathcal{X})) +1$. Let
  $\partial_i$ denote the $i$th boundary operator of $\nerv(\mathcal{X})$. By
  definition, $\beta_{d-1}$ is the dimension of the quotient of the
  vector space $\ker
  \partial_{d-1}$  by the vector space $\im \partial_{d}$. Now, $\ker
  \partial_{d-1}$ is contained in the space
  $\chain_{d-1}(\nerv(\mathcal{X}))$ spanned by the $(d-1)$-simplices of
  $\nerv(\mathcal{X})$; since $\nerv(\mathcal{X})$ has $n$ vertices it has at most
  $\binom{n}{d}$ simplices of dimension $d-1$ and $\ker
  \partial_{d-1}$ therefore has dimension at most $\binom{n}{d}$.
  This dimension can only go down by taking the quotient by $\im
  \partial_{d}$, so $\beta_{d-1}(\nerv(\mathcal{X})) \le
  \binom{n}{d}$.
\end{proof}

\paragraph{The number of holes in the union of $\F$.}

We now prove Theorem~\ref{thrm:cubic voids} by using
Lemma~\ref{l:hole-betti}.
\begin{proof}[Proof of Theorem~\ref{thrm:cubic voids}]
  Kovalev~\cite{k-svm-88} already established that any union of $n$
  convex objects in $\R^3$ has $O(n^3)$ holes.  Hence this bound applies to families of translates. 
  It remains to prove that this bound is tight by constructing a family
  whose union has $\Omega(n^3)$ holes. Let $K$ denote the convex body,
  and let $\F$ denote the family of $n=3m$ translates of $K$
  constructed above. Let $U = \bigcup_{X \in \F} X$.

  Recall that the maximal simplices of $\nerv(\F)$ are identified by
  Lemma~\ref{l:nerve}. By Lemma~\ref{l:hole-betti}, the number of
  holes of $U$ is $\beta_2(\nerv(\F)) +1$ which equals, by
  Equation~\eqref{e:betti}, $ \dim \ker \partial_2 - \rank
  \partial_3 +1$, where $\partial_i$ denotes the $i$th boundary map of
  $\nerv(\F)$. We compute $\beta_2(\nerv(\F))$ by computing explicitly
  a basis for $\ker \partial_2$ and a basis for $\im \partial_3$.

  To compute a basis of $\im \partial_3$, let $\mc{S}$ denote
  the set of $3$-simplices of $\nerv(\F)$ containing $B_1$ and let
  $\mc{T}$ stand for the set of images of the simplices of $\mc{S}$
  under $\partial_3$:
  \[ 
  \mc{S} = \{ \sigma : |\sigma| = 4 \hbox{ and } B_1 \in \sigma\}
  \quad \hbox{and} \quad \mc{T} = \{\partial_3 \sigma : \sigma \in
  \mc{S}\}.
  \]
  Observe that $\mc{T}$ is a linearly independent family. Indeed, for
  any $\sigma \in \mc{S}$, the $2$-simplex $\sigma \setminus \{B_1\}$
  has non-zero coefficient in $\partial_3(\sigma)$ but has zero
  coefficient in $\partial_3(\tau)$ for every $\tau \in \mc{S}
  \setminus \{\sigma\}$. To see that $\mc{T}$ spans $\im
  \partial_3$, let $\sigma$ be a $3$-simplex of $\nerv(\F)$.  If $B_1
  \in \sigma$ then $\sigma \in \mc{A}$ and so $\partial_3(\sigma) \in
  \mc{T}$. If $B_1 \notin \sigma$, since $\partial_3 \circ \partial_4
  =0$ we have
  \[ 
  \partial_3 \circ \partial_4 (\sigma \cup \{B_1\}) = \lambda
  \partial_3(\sigma) + \sum_{X \in \sigma} \lambda_v\partial_3(\sigma
  \cup \{B_1\} \setminus \{X\}) =0
  \]
  where $\lambda$ and the $\lambda_v$ are in $\{\pm1\}$. This implies
  that $\partial_3(\sigma)$ is a sum of $\pm\partial_3(\tau)$ with
  $\tau \in \mc{S}$, and thus lies in the span of $\mc{T}$.
  Therefore, as claimed, $\mc{T}$ is a basis of $\im \partial_3$.

  Now, let $\mc{S}'$ denote the set of 2-simplices in $\F$ that
  contain $B_1$ and let $\mc{T}' = \{\partial_2 \sigma : \sigma \in
  \mc{S}'\}$. The same arguments yield that $\mc{T}'$ is a basis of
  $\im \partial_2$.  Also let $\mc{S}''$ denote the set of all
  2-simplices contained in $\F$.

  We can finally compute $\beta_2(\nerv(\F))$ using the rank-nullity theorem,
  \begin{align*}
    \beta_2(\nerv(\F)) & = \nullity \partial_2 - \rank \partial_3 \\
    & = \dim \chain_2(\nerv(\F)) - \rank \partial_2 - \rank \partial_3 \\
    & = |\mc{S}''| -|\mc{T}'| -|\mc{T}|.
  \end{align*}
  Counting all quadruples in $\Delta_1$ or $\Delta_2$ that
  contain $B_1$ (taking care that some quadruples appear both in
  $\Delta_1$ and $\Delta_2$) we have
  \[
  | \mc{T} | = |\mc{S}| = 2\tbinom{2m-1}{3} - \tbinom{m-1}{3}.
  \]
  Next, counting all triples in $\F$ that contain $B_1$ we have
  \[
  | \mc{T}' | = |\mc{S}'| = \tbinom{3m-1}{2}.
  \]
  Then, counting all triples contained in $\nerv(\F)$, which are the triples in $\Delta_{i,j,k}$ plus $\Delta_{j,k}$ plus the triples in $\Delta_1$ or $\Delta_2$ (accounting for triples belonging to
  both $\Delta_1$ and $\Delta_2$), we have
  \[
  |\mc{S}''| = m^3 +m(m-1) +2\tbinom{2m}{3} - \tbinom{m}{3}.
  \]
  Finally, using $\binom{a}{b} - \binom{a-1}{b} = \binom{a-1}{b-1}$, we have
  \begin{align*}
    \beta_2(\nerv(\F)) & = |\mc{S}''| -|\mc{T}| -|\mc{T}'|
    \\
    & = m +m(m-1) +2\tbinom{2m}{3} - \tbinom{m}{3} 
    - 2\tbinom{2m-1}{3} + \tbinom{m-1}{3}
    - \tbinom{3m-1}{2}
    \\
    & = m^3 + m(m-1) + 2\tbinom{2m-1}{2} - \tbinom{m-1}{2} -  \tbinom{3m-1}{2}
    \\
    & = m^3 + m(m-1) +(2m-1)(2m-2) -\half(m-1)(m-2) -\half(3m-1)(3m-2)
    \\
    & = m^3 -m .
  \end{align*}
  So $\F$ is a family of $3m$ translates of a convex body in $\R^3$, and
  the union has $m^3 -m +1$ holes.
\end{proof}

\bibliographystyle{plain}
\bibliography{voids}

\clearpage

\appendix

\section{Witness points for Lemma~\ref{l:nerve}}
\label{apx:universal}

In this appendix we argue that the four sets of translates of~$K$
considered in Lemma~\ref{l:nerve} have non-empty intersection.  We
show this by exhibiting an explicit point in the intersection of each
set, in the following four claims.  We will need the bounds
\[ 
1 < \zeta_3 < \nicefrac{5}{4}, \quad \nicefrac{3}{2} < \zeta_2 <
\nicefrac{7}{4}. 
\]
\begin{claim}
  For $i \in [m]$ and $k \in [m-1]$ we have $\vec{w}_{0,k+1}
  +\vec{a}_{i} \in A_{i} \cap C_{k} \cap C_{k+1}$.
\end{claim}
\begin{proof}
  Since $w_{0,k+1} \in K$, $w_{0,k+1} + a_{i} \in A_{i}$.
  Furthermore, $w_{0,k+1}$ coincides with the point $\ut - w_{0,k}$
  of~$C_{k}$ and the point $\ut-w_{0,k+1}$ of $C_{k+1}$.  Since the
  vertical segment below $w_{0,k+1}$ lies in $C_{k} \cap C_{k+1}$ and
  $\eps < \zeta_3$, the claim holds.
\end{proof}

\begin{claim}
  For $(i,j,k) \in [m]^{3}$ we have
  $\vec{v}_{j,k} +\vec{a}_i +\vec{b}_j \in 
  A_{i} \cap B_{j} \cap C_{k}$.
\end{claim}
\begin{proof}
  The point $v_{j,k}+b_{j}$ lies in~$B_{j}$ and on the edge~$E_{k}$,
  the top edge of the facet~$R_{k}$ of~$C_{k}$. The translation $a_{i}$ is a
  small vertical translation, so $v_{j,k} + b_{j} + a_{i}$ lies
  inside~$R_{k}$ and therefore in~$C_{k}$, and lies in the interior of~$B_{j}$.
  Finally, since $v_{j,k} + b_{j} \in E_k \subset K$, the point
  $v_{j,k} + b_{j} + a_{i} \in A_{i}$.
\end{proof}

\begin{claim}
  We have $(0,-1,0) \in \bigcap \{A_1,\dots,A_m,B_1,\dots,B_m\}$.
\end{claim}
\begin{proof}
  The point $p = (0, -1, 0)$ lies on the bottom edge of~$K$
  connecting~$w_{0,0} = (0,0,0)$ and $u_{0} - w_{0,\infty} =
  (0,-2-\zeta_3,0)$. Since the vertical segment of length one with
  bottom end~$p$ lies in~$K$ and $\eps<1$, we have $p - a_{i} \in K$
  and therefore~$p \in A_{i}$, for $i \in [m]$.

  The length of the top edge of~$K$ connecting $w_{\infty,0} =
  (0,-\zeta_2,\zeta_3)$ and $u_{1} - w_{0,\infty} =
  (0,-2-\zeta_3,\zeta_3)$ is $2 + \zeta_3 - \zeta_2 > 1$, and so the
  path $p + \eta$ lies in~$K$.  This implies that $p - b_{j} = p +
  w_{j,0} \in K$, and therefore $p \in B_{j}$, for $j \in [m]$.
\end{proof}

\begin{claim}
  We have $p = (\zeta_2 , 2, -1) \in \bigcap
  \{B_1,\dots,B_m,C_1,\dots,C_m\}$.
\end{claim}
\begin{proof}
  For $j \in [m]$, we have $-b_{j} = w_{j,0} \in \{0\} \times
  [-\zeta_2,-1] \times [1,\zeta_3]$.  
  This gives 
  \[
  p - b_{j} \in \{\zeta_2\} \times [2-\zeta_2,1] \times
  [0, \zeta_3 -1] \subset \{\zeta_2\} \times [-2,1] \times [0, \zeta_3
    - 1].
  \]
  The four corners of this rectangle are 
  \begin{align*}
    \uz & = (\zeta_2, -2, 0) \in K, & q_{1} & = (\zeta_2, -2, \zeta_3 - 1)
    \in [\uz,\ut], \\
    q_{2} & = (\zeta_2, 1, 0) \in [\uz,w_{0,\infty}], & 
    q_{3} & = (\zeta_2,1,\zeta_3-1).
  \end{align*}
  Since $\vec{w}_{0,\infty} = (\zeta_2,\zeta_3,0)$,
  $\vec{w}_{1,\infty} = (\zeta_2,\zeta_3-1,1)$, and ${0 < \zeta_3-1 <
    1}$, we have
  \[
  q_{3} = (2-\zeta_3) \vec{w}_{0,\infty} + (\zeta_3-1)
  \vec{w}_{1,\infty} \in K,
  \]
  and so the entire rectangle lies in~$K$.  This implies $p - b_{j}
  \in K$, and so $p \in K + b_{j} = B_{j}$.

  For $k \in [m]$, we have
  \begin{align*}
    p-c_{k} & = p + \ut - w_{0,k} - w_{0,k+1} = 
    (2\zeta_2, 0, \zeta_3-1) - w_{0,k} - w_{0,k+1} \\
    & = \textstyle
    \Big(2\zeta_2 - 2\sum_{t = 1}^{k} t^{-2} - (k+1)^{-2}, 
    -2\sum_{t=1}^{k}t^{-3} - (k+1)^{-3}, \zeta_3-1\Big) \\
    & \in [0, 2\zeta_2 -2 - \nicefrac{1}{4}]
    \times [-2\zeta_3, -2 - \nicefrac{1}{8}]
    \times \{\zeta_3-1\} \\
    & \subset [0, \zeta_2 - \nicefrac{1}{2}]
    \times [-\nicefrac{5}{2}, -1]
    \times \{\zeta_3-1\}.
  \end{align*}
  Here we used $\zeta_2 < 7/4$ and $\zeta_3 < 5/4$.  The four corners
  of this rectangle are
  \begin{align*}
    q_{1} & = (0, -\nicefrac{5}{2}, \zeta_3-1), &
    q_{2} & = (0, -1, \zeta_3-1), \\
    q_{3} & = (\zeta_2 - \nicefrac{1}{2}, -\nicefrac{5}{2}, \zeta_3-1), &
    q_{4} & = (\zeta_2 - \nicefrac{1}{2}, -1, \zeta_3-1). 
  \end{align*}
  We will argue that all four corners lie in~$K$, implying that $p -
  c_{k} \in K$, and so $p \in K + c_{k} = C_{k}$.

  For $q_{1}$ and $q_{2}$ this is obvious, as they lie inside the
  convex hull of the points
  \begin{align*}
    w_{0,0} & = (0,0,0), & w_{1,0} & = (0, -1, 1), \\
    \uz + w_{0,\infty} & = (0, -2-\zeta_3,0), &
    \ut + w_{0, \infty} & = (0, -2-\zeta_3, \zeta_3).
  \end{align*}
  The point $q_{4}$ lies on the segment $q_{2}q_{5}$, where $q_{5} =
  (\zeta_2, -1, \zeta_3 - 1)$, and $q_{5}$ lies in the convex hull of
  $w_{0,\infty}$, $w_{1,\infty}$, $\uz$, and~$\ut$.

  Finally, the point $q_{3}$ lies on the facet $\uz\ut - E_0$ of~$K$.
  Indeed, the corners of this facet are
  \begin{align*}
    \uz & = (\zeta_2, -2, 0), & \uz - w_{0,1} & = (\zeta_2 -1, -3, 0), \\
    \ut & = (\zeta_2, -2, \zeta_3), & \ut - w_{0,1} & = (\zeta_2 -1,
    -3, \zeta_3). \qedhere
  \end{align*}
\end{proof}
\end{document}

%% file: kandslices.tex
\begin{tikzpicture}
\begin{yzxcoords}

\fill[black!10]
(5.5cm,1.5cm) coordinate (o) circle (1cm)
;

\path
(o)
+(0,0,.4) node (x) {\ns$x$}
+(.5,0,0) node (y) {\ns$y$}
+(0,.4,0) node (z) {\ns$z$}
;

\draw[thin,->] (o) -- (x);
\draw[thin,->] (o) -- (y);
\draw[thin,->] (o) -- (z);

\begin{scope}
\clip
(5.5cm,-3.2cm) coordinate (o) ++(-.3cm,-.3cm) rectangle ++(1.8cm,1.8cm)
;
\fill[black!10]
(o) circle (1cm)
;
\end{scope}

\path
(o)
+(.7cm,0cm) node (x) {\ns$x$}
+(0cm,.7cm) node (z) {\ns$z$}
;

\draw[thin,->] (o) -- (x);
\draw[thin,->] (o) -- (z);

\matrix[column sep=.4cm]{

\begin{scope}[scale=.7]
\polytopeB{very thin, black!20}
\end{scope}

\node[above right, shift={(-2pt,-2pt)}] at (b1) {$\vec{a}$};

\draw[fill=black!15,fill opacity=.5]
(a1) -- (v) -- (a2) -- cycle
;

\draw[fill=black!30,fill opacity=.5]
(a2) -- (v) -- (b2) -- cycle
;

\draw[fill=black!25,fill opacity=.5]
(a1) -- (a2) -- (a3) -- cycle
;

\draw[fill=black!45,fill opacity=.5]
(b1) -- (v) -- (b2) -- cycle
;

\draw[fill=black!70,fill opacity=.5]
(a2) -- (b2) -- (b3) -- (a3) -- cycle
;

\node[left] at (a1) {$\vec{b}$};
\node[left] at (a3) {$\vec{d}$};
\node[below] at (b3) {$\vec{c}$};
\node[above] at (v) {$\vec{e}$};
\node[right] at (b2) {$\vec{f}$};
\node[left, shift={(-2pt,0pt)}] at (a2) {$\vec{g}$};

&

\begin{scope}[scale=.7]
\polytopeB{thin, gray}

\draw[ultra thick, green!50!black, fill=green!50!black, fill opacity=0.1]
(a3) -- (a1) -- (b1) -- (b3) -- (a3)
;

\draw[ultra thick, blue, fill=blue, fill opacity=0.1]
(1,-1,2) -- (1,0,2) -- (v) -- (1,0,0) -- (1,-1,0) -- (1,-1,2)
;

\draw[ultra thick, red]
(a2) -- (b2)
;
\end{scope}

\\
};

\end{yzxcoords}

\begin{scope}[shift={(-2.8cm,-4.5cm)}]

\draw[thick, red]
(-1,1.05) -- (7,1.05) coordinate (f1)
(-1,.8) -- (7,.8) coordinate (f2)
(-1,.55) -- (7,.55) coordinate (f3)
;

\fill[blue!50!gray!20]
(-.3,0) -- (1,1.3) -- (2.3,0)
(1.7,0) -- (3,1.3) -- (4.3,0)
(3.7,0) -- (5,1.3) -- (6.3,0)
;

\draw[thick, green!50!black]
(-1,0) -- (-1,1.3) coordinate (a) -- (7,1.3) -- (7,0)
;

\draw[thick, blue]
(-.3,0) -- (1,1.3) -- (2.3,0)
(1.7,0) -- (3,1.3) -- (4.3,0)
(3.7,0) -- (5,1.3) -- (6.3,0)
;

\path (-1,.3) coordinate (c0);
\node[left] at (c0) {\ns $\vec{c}'$};
\node[left] at (-1,1.3) {\ns $\vec{a}$};
\node[right] at (f1) {\ns $A_1$};
\node[right] at (f2) {\ns $A_2$};
\node[right] at (f3) {\ns $A_3$};
\node[above] at (1,1.3) {\ns $B_1$};
\node[above] at (3,1.3) {\ns $B_2$};
\node[above] at (5,1.3) {\ns $B_3$};
\node[right] at (7,1.3) {\ns $\vec{b}$};

\draw [decorate,decoration={brace,amplitude=5pt,mirror}]
(a) -- (c0) node [midway,xshift=-9pt] {\ns $\ell$}
;

\draw[thin]
(c0) -- ++(.15,0);
\draw[thin, dotted]
(c0) -- (2,.3);

\end{scope}

\end{tikzpicture}

%% file: unionB.tex
\begin{tikzpicture}
\begin{yzxcoords}

\begin{scope}[scale=.8]

\draw[very thin,black!20]
(0,0,0) coordinate (b1) -- (2,0,0) coordinate (b2) -- (0,-2,0) coordinate (b3) -- cycle
(1,1,1) coordinate (v0)
\foreach \i in {1,2}
{ -- ++(.4,-.4,.4) coordinate (x\i) -- ++(-.4,.4,.4) coordinate (v\i) }
(v0) 
\foreach \i in {1,2}
{ -- ++(-.4,-.4,.4) coordinate (y\i) -- (v\i) }
++(-1,-1,1) coordinate (a1) -- +(2,0,0) coordinate (a2) -- ++(0,-2,0) coordinate (a3) -- cycle
(a1) -- (v2) -- (a2)
(b1) -- (v0) -- (b2)
\foreach \i in {1,...,3}
{ (a\i) -- (b\i) }
;

\draw[fill=black!15,fill opacity=.5]
(a1) -- (v2) -- (a2) -- cycle
(y2) -- (v1) -- (x2) -- cycle
(y1) -- (v0) -- (x1) -- cycle
;

\draw[fill=black!45,fill opacity=.5]
(y2) -- (v2) -- (x2) -- cycle
(y1) -- (v1) -- (x1) -- cycle
(b1) -- (v0) -- (b2) -- cycle
;

\draw[fill=black!30,fill opacity=.5]
(a2) -- (v2) -- (x2) -- (v1) -- (x1) -- (v0) -- (b2) -- cycle
;

\draw[fill=black!25,fill opacity=.5]
(a1) -- (a2) -- (a3) -- cycle
;

\draw[fill=black!70,fill opacity=.5]
(a2) -- (b2) -- (b3) -- (a3) -- cycle
;


\node[above] at (v0) {\ns $B_3$};
\node[above] at (v1) {\ns $B_2$};
\node[above] at (v2) {\ns $B_1$};

\end{scope}
\end{yzxcoords}
\end{tikzpicture}

%% file: unionA.tex
\begin{tikzpicture}
\begin{yzxcoords}
\begin{scope}[scale=.8]

\draw[thin]
(0,0,2) coordinate (a1)
 -- (2,0,2) coordinate (c1) -- ++(-.6,-.6,0) coordinate (e1)
 -- (2,-.6,2) coordinate (c2) -- ++(-.6,-.6,0) coordinate (e2)
 -- (2,-1.2,2) coordinate (c3)
 -- (0,-3.2,2) coordinate (a3) -- cycle
(0,0,0) coordinate (b1)
 -- (2,0,0) coordinate (d1) -- ++(-.6,-.6,0) coordinate (f1)
 -- (2,-.6,0) coordinate (d2) -- ++(-.6,-.6,0) coordinate (f2)
 -- (2,-1.2,0) coordinate (d3)
 -- (0,-3.2,0) coordinate (b3)
(c2) -- ++(-.3,.3,-.3) coordinate (x1) -- (e1)
(c3) -- ++(-.3,.3,-.3) coordinate (x2) -- (e2)
(d2) -- ++(-.3,.3,.3) coordinate (y1) -- (f1)
(d3) -- ++(-.3,.3,.3) coordinate (y2) -- (f2)
(x1) -- (y1)
(x2) -- (y2)
(a1) -- (1,1,1) coordinate (v) -- (b1)
(c1) -- (v) -- (d1) -- cycle
(c2) -- (d2)
(c3) -- (d3)
(a3) -- (b3)
;

\draw[very thin,black!20]
(a1) -- (b1) -- (b3)
;

\draw[fill=black!15,fill opacity=0.5]
(a1) -- (c1) -- (v) -- cycle
(e1) -- (c2) -- (x1) -- cycle
(e2) -- (c3) -- (x2) -- cycle
;

\draw[fill=black!30,fill opacity=0.5]
(c1) -- (d1) -- (v) -- cycle
(c2) -- (d2) -- (y1) -- (x1) -- cycle
(c3) -- (d3) -- (y2) -- (x2) -- cycle
;

\draw[fill=black!25,fill opacity=0.5]
(a1) -- (c1) -- (e1) -- (c2) -- (e2) -- (c3) -- (a3)
;

\draw[fill=black!70,fill opacity=0.5]
(a3) -- (c3) -- (d3) -- (b3) -- cycle
(e2) -- (c2) -- (d2) -- (f2) -- (y2) -- (x2) -- cycle
(e1) -- (c1) -- (d1) -- (f1) -- (y1) -- (x1) -- cycle
;

\draw[fill=black!45,fill opacity=0.5]
(b1) -- (d1) -- (v) -- cycle
(f2) -- (d3) -- (y2) -- cycle
(f1) -- (d2) -- (y1) -- cycle
;


\node[right] at (d1) {\ns $A_1$};
\node[right] at (d2) {\ns $A_2$};
\node[right] at (d3) {\ns $A_3$};

\end{scope}
\end{yzxcoords}
\end{tikzpicture}

%% file: front.tex
\begin{tikzpicture}
\begin{yxzcoords}

\fill[black!10]
(-2cm,3.5cm) coordinate (o) circle (1cm)
;

\path
(o)
+(0,0,.4) node (z) {\ns$z$}
+(.5,0,0) node (y) {\ns$y$}
+(0,.4,0) node (x) {\ns$x$}
;

\draw[thin,->] (o) -- (x);
\draw[thin,->] (o) -- (y);
\draw[thin,->] (o) -- (z);

\begin{scope}[scale=1.5]

\input{common.tex}

\path
(w0_0) ++(-.45,0,0) coordinate (a)
;

\matrix [column sep=.1cm, row sep=.5cm]{




\draw[very thin, black!30]
(w5_0) -- (w0_0) -- (a) -- (w0_3)
(w0_1) ++(-1,0,0) coordinate (b)
(w0_1) -- (intersection of b--w0_1 and a--w0_3)
(w0_2) ++(-1,0,0) coordinate (b)
(w0_2) -- (intersection of b--w0_2 and a--w0_3)
\foreach \x in {1,...,4}
{ 
(w\x_0) ++(-1,0,0) coordinate (b)
(w\x_0) -- (intersection of w0_0--w5_0 and b--w\x_0) 
}
;

\draw[thick] 
(w0_0) node[scale=.5] {$\bullet$}
\foreach \y in {1,...,3}
{ -- (w0_\y) node[scale=.5] {$\bullet$} }
(w0_0) 
\foreach \x in {1,...,5}
{ -- (w\x_0) node[scale=.5] {$\bullet$} }
;


\node[below] at (w0_0) {\ns $\vec{w}_{0,0}$};

\foreach \y in {1,...,2}
{\node[right] at (w0_\y) {\ns $\vec{w}_{\y,0}$};}
{\node[right] at (w0_3) {\ns $\vec{w}_{m,0}$};}

\foreach \x in {1,...,4}
{\node[below right,shift={(-3pt,1pt)}] at (w\x_0) {\ns $\vec{w}_{0,\x}$};}
{\node[below right,shift={(-3pt,1pt)}] at (w5_0) {\ns $\vec{w}_{0,m+2}$};}

\node[shift={(1cm,-.6cm)}] at (w0_3) {$\eta$};
\node[shift={(.3cm,.2cm)}] at (w5_0) {$\gamma$};

&

%

\foreach \y in {0,...,3}
{
 \draw[thick]
  (w0_\y)
  \foreach \x in {0,...,5}
   { -- (w\x_\y) }
 ;
}

\foreach \x in {0,...,5}
{
 \draw[thick]
  (w\x_0)
  \foreach \y in {0,...,3}
   { -- (w\x_\y) }
 ;
}

\node[below] at (w0_0) {\ns $\vec{w}_{0,0}$};
\node[above,shift={(0pt,-2pt)}] at (w0_3) {\ns $\vec{w}_{m,0}$};
\node[right,shift={(-1pt,-3pt)}] at (w5_0) {\ns $\vec{w}_{0,m+2}$};
\node[shift={(0pt,3pt)}] at (w5_3) {\ns $\vec{w}_{m,m+2}$};


\\

\foreach \y in {0,...,3}
{
 \draw[thick]
  (w0_\y)
  \foreach \x in {0,...,5}
   { -- (w\x_\y) }
 ;
}

\foreach \x in {0,...,5}
{
 \draw[thick]
  (w\x_0)
  \foreach \y in {0,...,3}
   { -- (w\x_\y) }
 ;
}

\foreach \x in {1,...,5}
{
\foreach \y in {0,...,3}
{
\node[scale=.8] at (v\x_\y) {\color{blue} $\bullet$};
}}

\node[scale=.7] at (w0_1) {$\bullet$};
\node[scale=.7] at (w1_1) {$\bullet$};

\node[shift={(-10pt,0pt)}] at (w0_1) {\ns $\vec{w}_{1,0}$};
\node[shift={(-8pt,-4pt)}] at (w1_1) {\ns $\vec{w}_{1,1}$};
\node[blue,shift={(3pt,5pt)}] at (v1_1) {\ns $\vec{v}_{1,0}$};


& 

%

\path 
(v5_0) coordinate (b);
\foreach \x in {5,...,2}
{
\path 
(b) coordinate (b\x)
(v\x_0) coordinate (a\x) coordinate (b);
}

\path
(v5_1) coordinate (b);
\foreach \x in {5,...,2}
{
  \draw[thick]
  (a\x) -- (v\x_1) -- (b\x) -- cycle;
  \path 
  (b) coordinate (b\x)
  (v\x_1) coordinate (a\x) coordinate (b);
}

\path
(v1_1) coordinate (a1)
(v2_1) coordinate (b1);

\foreach \y in {2,...,3}
{
 \path (v5_\y) coordinate (b); 
 \foreach \x in {5,...,1}
 {
  \draw[thick]
  (a\x) -- (v\x_\y) -- (b\x) -- cycle;
  \path 
  (b) coordinate (b\x)
  (v\x_\y) coordinate (a\x) coordinate (b);
 }
}

\draw[thick]
(v1_0) -- (v2_0) -- (v1_1)
(v1_3)
\foreach \x in {2,...,5}
 { -- (v\x_3)}
;

\draw[thick]
(w0_0)
\foreach \y in {1,...,3}
{ -- (w0_\y)}
;

\draw[thick]
(w5_0)
\foreach \y in {1,...,3}
{ -- (w5_\y)}
;

\foreach \y in {0,...,3}
{
\draw[thick]
(w0_\y) -- (v1_\y)
(w5_\y) -- (v5_\y)
;
}

\\
};


\end{scope}
\end{yxzcoords}
\end{tikzpicture}

%% file: gammas.tex
\begin{tikzpicture}
\clip
(-4cm,15mm) rectangle (5cm,7cm)
;

\begin{scope}
\clip
(4cm,6cm) coordinate (o) ++(-.3cm,-.3cm) rectangle ++(1.8cm,1.8cm)
;
\fill[black!10]
(o) circle (1cm)
;
\end{scope}

\path
(o)
+(.7cm,0cm) node (x) {\ns$x$}
+(0cm,.7cm) node (y) {\ns$y$}
;

\draw[thin,->] (o) -- (x);
\draw[thin,->] (o) -- (y);

\begin{scope}[z={(2pt,2pt)}, x={(0cm,7cm)}, y={(5cm,0cm)}]

\input{common.tex}

\path 
(v5_0) coordinate (b);
\foreach \x in {5,...,2}
{
\path 
(b) coordinate (b\x)
(v\x_0) coordinate (a\x) coordinate (b);
}

\path
(v1_0) coordinate (a1)
(v2_0) coordinate (b1);

\foreach \y in {1,...,3}
{
 \path (v5_\y) coordinate (b); 
 \foreach \x in {5,...,1}
 {
  \draw[very thin,black!15]
  (a\x) -- (v\x_\y) -- (b\x);
  \path 
  (b) coordinate (b\x)
  (v\x_\y) coordinate (a\x) coordinate (b);
 }
}

\draw[thick]
(w0_0)
\foreach \x in {1,...,5}
{--(w\x_0)}
;

\draw[very thick, black!35]
\foreach \y in {1,2,3}
{
(w0_\y)
\foreach \x in {1,...,5}
{--(w\x_\y)}
};

\draw[blue]
\foreach \y in {1,2,3}
{(w0_\y)
\foreach \x in {1,...,5}
 { -- (v\x_\y)}
 --(w5_\y)
};

\foreach \x in {1,...,5}
{
\foreach \y in {0,...,3}
{
\node[scale=.8] at (v\x_\y) {\color{blue} $\bullet$};
}}

\node[above right] at (w0_0) {$\gamma$};
\node[above right,blue] at (w0_1) {$\hpaths_1$};
\node[above right,blue] at (w0_2) {$\hpaths_2$};
\node[above right,blue] at (w0_3) {$\hpaths_m$};
\end{scope}
\end{tikzpicture}

%% file: vert_cones_Km.tex
\begin{tikzpicture}
\clip
(-7cm,-2cm) rectangle (7cm,2cm)
;

\begin{yxzcoords}

\input{common.tex}


\path 
(v5_1) coordinate (b)
;
\foreach \x in {4,...,1}
{
\path 
(b) coordinate (b\x)
(v\x_1) coordinate (b)
;
}
\foreach \x in {4,...,1}
{
\path
(w\x_1) ++(-1,0,0) coordinate (a) 
(intersection of b\x--v\x_1 and w\x_1--a) coordinate (c\x)
;
\path
(w\x_1) ++(0,0,-1) coordinate (e)
(intersection of c\x--w\x_0 and w\x_1--e) 
coordinate (d\x)
;
}

\foreach \x in {0,...,5}
\path
(w\x_1) ++(0,0,-1) coordinate (e)
(w\x_0) ++(-1,0,0) coordinate (a)
(intersection of e--w\x_1 and a--w\x_0) coordinate (f\x)
;

\foreach \x in {0,...,5}
{
\path
(w\x_1) ++(0,0,1) coordinate (g\x)
;
}

\path
(w0_0) ++(-1,0,0) coordinate (h0)
++(0,0,.936) coordinate (h1)
(w5_0) ++(-1,0,.936) coordinate (h2)
;

\matrix [column sep=.5cm, row sep=.2cm]{

\path 
(v5_0) coordinate (b);
\foreach \x in {5,...,2}
{
\path 
(b) coordinate (b\x)
(v\x_0) coordinate (a\x) coordinate (b);
}

\path
(v5_1) coordinate (b);
\foreach \x in {5,...,2}
{
  \pgfmathsetmacro\c{.09*\x}
  \definecolor{currentcolor}{hsb}{0, 0, \c}
  \draw[fill=currentcolor,fill opacity=0.5]
  (a\x) -- (v\x_1) -- (b\x) -- cycle;
  \path 
  (b) coordinate (b\x)
  (v\x_1) coordinate (a\x) coordinate (b);
}

\path
(v1_1) coordinate (a1)
(v2_1) coordinate (b1);

\foreach \y in {2,...,3}
{
 \path (v5_\y) coordinate (b); 
 \foreach \x in {5,...,1}
 {
  \pgfmathsetmacro\c{.09*\x+.11*\y}
  \definecolor{currentcolor}{hsb}{0, 0, \c}
  \draw[fill=currentcolor,fill opacity=0.5]
  (a\x) -- (v\x_\y) -- (b\x) -- cycle;
  \path 
  (b) coordinate (b\x)
  (v\x_\y) coordinate (a\x) coordinate (b);
 }
}

\foreach \x in {2,...,5}
{
 \path
 (v\x_0) coordinate (a\x)
 ;
}

\foreach \y in {1,...,3}
{
 \path
 (v1_\y) coordinate (b)
 ;
 \foreach \x in {2,...,5}
 {
  \pgfmathsetmacro\c{.09*\x+.11*\y+.04}
  \definecolor{currentcolor}{hsb}{0, 0, \c}
  \draw[fill=currentcolor,fill opacity=.5]
  (b) -- (a\x) -- (v\x_\y) coordinate (a\x) coordinate (b) -- cycle
  ;
 }
}

\draw[fill=black,fill opacity=0.5]
(w0_0) -- (w0_1) -- (v1_1) -- (w1_0) -- cycle
;
\draw[fill=black!80,fill opacity=0.5]
(w0_1) -- (w0_2) -- (v1_2) -- (v1_1) -- cycle
;
\draw[fill=black!67,fill opacity=0.5]
(w0_2) -- (w0_3) -- (v1_3) -- (v1_2) -- cycle
;

\draw[fill=black!47,fill opacity=0.5]
(w5_0) -- (w5_1) -- (v5_1) -- (w4_0) -- cycle
;
\draw[fill=black!34,fill opacity=0.5]
(w5_1) -- (w5_2) -- (v5_2) -- (v5_1) -- cycle
;
\draw[fill=black!18,fill opacity=0.5]
(w5_2) -- (w5_3) -- (v5_3) -- (v5_2) -- cycle
;

\fill[black!85,fill opacity=0.5]
(w0_0)
\foreach \y in {1,...,3}
{ -- (w0_\y)}
-- (h1) -- (h0) -- cycle
;

\fill[black!16,fill opacity=0.5]
(v1_3)
\foreach \x in {2,...,5}
 { -- (v\x_3)}
-- (w5_3) -- (h2) -- (h1) -- (w0_3) -- cycle
;


&

\node at (0,0,0) {\Large $\cap$};

&



\path
(f5) coordinate (a)
(g5) coordinate (b)
;
\foreach \x in {4,...,0}
{
  \pgfmathsetmacro\c{.09*\x+.26}
  \definecolor{currentcolor}{hsb}{0, 0, \c}
  \draw[olive,fill=currentcolor,fill opacity=0.5]
  (a) -- (f\x) -- (g\x) -- (b) -- cycle;
  \path 
  (f\x) coordinate (a)
  (g\x) coordinate (b)
  ;
}

\draw[thick]
(w0_1)
\foreach \x in {1,...,5}
{-- (w\x_1)}
;

&

\node at (0,0,0) {\Large $=$};

&

\path
(w0_1) coordinate (d0)
(w5_1) coordinate (d5)
(w0_1) coordinate (a)
(f0) coordinate (b)
;
\foreach \x in {1,...,5}
{
  \pgfmathsetmacro\c{.2*\x}
  \definecolor{currentcolor}{hsb}{.66666, .4, \c}
  \draw[blue,thick,fill=currentcolor,fill opacity=0.5]
  (a) -- (v\x_1) -- (d\x) -- (f\x) -- (b) -- cycle;
  \path 
  (d\x) coordinate (a)
  (f\x) coordinate (b)
  ;
}

\foreach \x in {0,...,5}
{
\draw[olive,opacity=0.5]
(d\x) -- (g\x)
;
}

\draw[olive,opacity=0.5]
(g0)
\foreach \x in {1,...,5}
{-- (g\x)}
;

\path 
(v5_0) coordinate (b);
\foreach \x in {5,...,2}
{
\path 
(b) coordinate (b\x)
(v\x_0) coordinate (a\x) coordinate (b);
}

\path
(v5_1) coordinate (b);
\foreach \x in {5,...,2}
{
\draw[very thin,gray,opacity=0.5]
  (a\x) -- (v\x_1) -- (b\x) -- cycle;
  \path 
  (b) coordinate (b\x)
  (v\x_1) coordinate (a\x) coordinate (b);
}

\path
(v1_1) coordinate (a1)
(v2_1) coordinate (b1);

\foreach \y in {2,...,3}
{
 \path (v5_\y) coordinate (b); 
 \foreach \x in {5,...,1}
 {
\draw[very thin,gray,opacity=0.5]
  (a\x) -- (v\x_\y) -- (b\x) -- cycle;
  \path 
  (b) coordinate (b\x)
  (v\x_\y) coordinate (a\x) coordinate (b);
 }
}

\draw[very thin,gray,opacity=0.5]
(v1_0) -- (v2_0) -- (v1_1)
(v1_3)
\foreach \x in {2,...,5}
 { -- (v\x_3)}
;

\draw[very thin,gray,opacity=0.5]
(w0_0)
\foreach \y in {1,...,3}
{ -- (w0_\y)}
;

\draw[very thin,gray,opacity=0.5]
(w5_0)
\foreach \y in {1,...,3}
{ -- (w5_\y)}
;

\foreach \y in {0,...,3}
{
\draw[very thin,gray,opacity=0.5]
(w0_\y) -- (v1_\y)
(w5_\y) -- (v5_\y)
;
}

\\
};

\end{yxzcoords}
\end{tikzpicture}

%




%% file: back.tex
\begin{tikzpicture}
\begin{yxzcoords}

\fill[black!10]
(3.5cm,2cm) coordinate (o) circle (1cm)
;

\path
(o)
+(0,0,.4) node (z) {\ns$z$}
+(.5,0,0) node (y) {\ns$y$}
+(0,.4,0) node (x) {\ns$x$}
;

\draw[thin,->] (o) -- (x);
\draw[thin,->] (o) -- (y);
\draw[thin,->] (o) -- (z);

\begin{scope}[scale=1.5]

\input{common.tex}

\draw[thick]
(w0_0)
\foreach \x in {1,...,5}
 { -- (w\x_0) }
;

\draw[thick]
(s0)
\foreach \x in {1,...,5}
 { -- (s\x) }
;

\draw[thick]
(t0)
\foreach \x in {1,...,5}
 { -- (t\x) }
;

\foreach \x in {0,...,5}
\draw[thick]
 (s\x) -- (t\x);

\draw[ultra thick,green!50!black]
(t0)
\foreach \x in {1,...,5}
{ -- (t\x) }
;

\draw[ultra thick,green!50!black]
(s0)
\foreach \x in {1,...,5}
{ -- (s\x) }
;

\node at (1.2,0,0) {$\gamma$};
\node at (-1,0,1.05) {$\ut-\gamma$};
\node at (-0.8,0,-0.05) {$\uz-\gamma$};
\node[right] at (t0) {$\ut$};
\node[right] at (s0) {$\uz$};

\end{scope}
\end{yxzcoords}
\end{tikzpicture}

%% file: kex2.tex
\begin{tikzpicture}
\begin{yxzcoords}

\fill[black!10]
(-6cm,-2.5cm) coordinate (o) circle (1cm)
;

\path
(o)
+(0,0,.4) node (z) {\ns$z$}
+(.5,0,0) node (y) {\ns$y$}
+(0,.4,0) node (x) {\ns$x$}
;

\draw[thin,->] (o) -- (x);
\draw[thin,->] (o) -- (y);
\draw[thin,->] (o) -- (z);

\begin{scope}[scale=1.5]

\input{common.tex}

\matrix [column sep=.3cm, row sep=.5cm]{

\draw[very thin, black!30]
(w5_0) -- (s0) -- (t0) -- (w5_3)
(s4) -- (s5)
(t4) -- (t5)
(s0) -- (s1)
(t0) -- (t1)
;

\draw[very thin, black!30]
(s1) -- (s2) -- (s3) -- (s4)
(s1) -- (t1)
(s2) -- (t2)
(s3) -- (t3)
(s4) -- (t4)
;

\path 
(v5_0) coordinate (b);
\foreach \x in {5,...,2}
{
\path 
(b) coordinate (b\x)
(v\x_0) coordinate (a\x) coordinate (b);
}

\path
(v5_1) coordinate (b);
\foreach \x in {5,...,2}
{
  \pgfmathsetmacro\c{.09*\x}
  \definecolor{currentcolor}{hsb}{0, 0, \c}
  \draw[fill=currentcolor,fill opacity=0.5]
  (a\x) -- (v\x_1) -- (b\x) -- cycle;
  \path 
  (b) coordinate (b\x)
  (v\x_1) coordinate (a\x) coordinate (b);
}

\path
(v1_1) coordinate (a1)
(v2_1) coordinate (b1);

\foreach \y in {2,...,3}
{
 \path (v5_\y) coordinate (b); 
 \foreach \x in {5,...,1}
 {
  \pgfmathsetmacro\c{.09*\x+.11*\y}
  \definecolor{currentcolor}{hsb}{0, 0, \c}
  \draw[fill=currentcolor,fill opacity=0.5]
  (a\x) -- (v\x_\y) -- (b\x) -- cycle;
  \path 
  (b) coordinate (b\x)
  (v\x_\y) coordinate (a\x) coordinate (b);
 }
}

\foreach \x in {2,...,5}
{
 \path
 (v\x_0) coordinate (a\x)
 ;
}

\foreach \y in {1,...,3}
{
 \path
 (v1_\y) coordinate (b)
 ;
 \foreach \x in {2,...,5}
 {
  \pgfmathsetmacro\c{.09*\x+.11*\y+.04}
  \definecolor{currentcolor}{hsb}{0, 0, \c}
  \draw[fill=currentcolor,fill opacity=.5]
  (b) -- (a\x) -- (v\x_\y) coordinate (a\x) coordinate (b) -- cycle
  ;
 }
}

\draw[fill=black,fill opacity=0.5]
(w0_0) -- (w0_1) -- (v1_1) -- (w1_0) -- cycle
;
\draw[fill=black!80,fill opacity=0.5]
(w0_1) -- (w0_2) -- (v1_2) -- (v1_1) -- cycle
;
\draw[fill=black!67,fill opacity=0.5]
(w0_2) -- (w0_3) -- (v1_3) -- (v1_2) -- cycle
;

\draw[fill=black!47,fill opacity=0.5]
(w5_0) -- (w5_1) -- (v5_1) -- (w4_0) -- cycle
;
\draw[fill=black!34,fill opacity=0.5]
(w5_1) -- (w5_2) -- (v5_2) -- (v5_1) -- cycle
;
\draw[fill=black!18,fill opacity=0.5]
(w5_2) -- (w5_3) -- (v5_3) -- (v5_2) -- cycle
;

\draw[fill=black!85,fill opacity=0.5]
(w0_0)
\foreach \y in {1,...,3}
{ -- (w0_\y)}
-- (t5) -- (s5) -- cycle
;

\draw[fill=black!16,fill opacity=0.5]
(v1_3)
\foreach \x in {2,...,5}
 { -- (v\x_3)}
-- (w5_3)
\foreach \x in {0,...,5}
 { -- (t\x)}
-- (w0_3) -- cycle
;

&

\path 
(v5_0) coordinate (b);
\foreach \x in {5,...,2}
{
\path 
(b) coordinate (b\x)
(v\x_0) coordinate (a\x) coordinate (b);
}

\path
(v5_1) coordinate (b);
\foreach \x in {5,...,2}
{
\draw[very thin, black!30]
  (a\x) -- (v\x_1) -- (b\x) -- cycle;
  \path 
  (b) coordinate (b\x)
  (v\x_1) coordinate (a\x) coordinate (b);
}

\path
(v1_1) coordinate (a1)
(v2_1) coordinate (b1);

\foreach \y in {2,...,3}
{
 \path (v5_\y) coordinate (b); 
 \foreach \x in {5,...,1}
 {
\draw[very thin, black!30]
  (a\x) -- (v\x_\y) -- (b\x) -- cycle;
  \path 
  (b) coordinate (b\x)
  (v\x_\y) coordinate (a\x) coordinate (b);
 }
}

\draw[very thin, black!30]
(v1_0) -- (v2_0) -- (v1_1)
(v1_3)
\foreach \x in {2,...,5}
 { -- (v\x_3)}
;

\draw[very thin, black!30]
(w0_0)
\foreach \y in {1,...,3}
{ -- (w0_\y)}
;

\draw[very thin, black!30]
(w5_0)
\foreach \y in {1,...,3}
{ -- (w5_\y)}
;

\foreach \y in {0,...,3}
{
\draw[very thin, black!30]
(w0_\y) -- (v1_\y)
(w5_\y) -- (v5_\y)
;
}

\draw[very thin, black!30]
(t5) -- (w0_3)
;

\path
(s5) coordinate (a)
(t5) coordinate (b)
;
\foreach \x in {4,...,0}
{
  \pgfmathsetmacro\c{.6-.09*\x}
  \definecolor{currentcolor}{hsb}{0, 0, \c}
  \draw[fill=currentcolor,fill opacity=0.5]
  (a) -- (s\x) -- (t\x) -- (b) -- cycle;
  \path 
  (s\x) coordinate (a)
  (t\x) coordinate (b)
  ;
}

\draw[fill=black, fill opacity=0.5]
(w0_0)
\foreach \x in {1,...,5}
 { -- (w\x_0)}
\foreach \x in {0,...,5}
 { -- (s\x)}
 -- cycle
;

\draw[fill=black!20,fill opacity=0.5]
(w5_0)
\foreach \y in {1,...,3}
{ -- (w5_\y)}
-- (t0) -- (s0) -- cycle
;

\\
};

\end{scope}
\end{yxzcoords}
\end{tikzpicture}

%% file: trans.tex
\begin{tikzpicture}[scale=1.5]
\clip
(-5cm,-6cm) rectangle (5cm,8cm)
;

\begin{yxzcoords}

\matrix [column sep=.5cm, row sep=0cm]{

\begin{scope}[scale=1.5]
\begin{scope}[shift={(0,0,-.1)}]

\input{common.tex}

\begin{scope}[opacity=0.2]

\path
(0,0,.1)
\foreach \x in {0,...,5}
{ +({50-\x*20}:1) coordinate (r\x) }
;

\draw[very thin, black!30]
(w5_0) -- (s0) -- (t0) -- (w5_3)
(s4) -- (s5)
(t4) -- (t5)
(s0) -- (s1)
(t0) -- (t1)
;

\draw[very thin, black!30]
(s1) -- (s2) -- (s3) -- (s4)
(s1) -- (t1)
(s2) -- (t2)
(s3) -- (t3)
(s4) -- (t4)
;

\path 
(v5_0) coordinate (b);
\foreach \x in {5,...,2}
{
\path 
(b) coordinate (b\x)
(v\x_0) coordinate (a\x) coordinate (b);
}

\path
(v5_1) coordinate (b);
\foreach \x in {5,...,2}
{
  \pgfmathsetmacro\c{.09*\x}
  \definecolor{currentcolor}{hsb}{0, 0, \c}
  \draw[fill=currentcolor,fill opacity=0.5]
  (a\x) -- (v\x_1) -- (b\x) -- cycle;
  \path 
  (b) coordinate (b\x)
  (v\x_1) coordinate (a\x) coordinate (b);
}

\path
(v1_1) coordinate (a1)
(v2_1) coordinate (b1);

\foreach \y in {2,...,3}
{
 \path (v5_\y) coordinate (b); 
 \foreach \x in {5,...,1}
 {
  \pgfmathsetmacro\c{.09*\x+.11*\y}
  \definecolor{currentcolor}{hsb}{0, 0, \c}
  \draw[fill=currentcolor,fill opacity=0.5]
  (a\x) -- (v\x_\y) -- (b\x) -- cycle;
  \path 
  (b) coordinate (b\x)
  (v\x_\y) coordinate (a\x) coordinate (b);
 }
}

\foreach \x in {2,...,5}
{
 \path
 (v\x_0) coordinate (a\x)
 ;
}

\foreach \y in {1,...,3}
{
 \path
 (v1_\y) coordinate (b)
 ;
 \foreach \x in {2,...,5}
 {
  \pgfmathsetmacro\c{.09*\x+.11*\y+.04}
  \definecolor{currentcolor}{hsb}{0, 0, \c}
  \draw[fill=currentcolor,fill opacity=.5]
  (b) -- (a\x) -- (v\x_\y) coordinate (a\x) coordinate (b) -- cycle
  ;
 }
}

\draw[fill=black,fill opacity=0.5]
(w0_0) -- (w0_1) -- (v1_1) -- (w1_0) -- cycle
;
\draw[fill=black!80,fill opacity=0.5]
(w0_1) -- (w0_2) -- (v1_2) -- (v1_1) -- cycle
;
\draw[fill=black!67,fill opacity=0.5]
(w0_2) -- (w0_3) -- (v1_3) -- (v1_2) -- cycle
;

\draw[fill=black!47,fill opacity=0.5]
(w5_0) -- (w5_1) -- (v5_1) -- (w4_0) -- cycle
;
\draw[fill=black!34,fill opacity=0.5]
(w5_1) -- (w5_2) -- (v5_2) -- (v5_1) -- cycle
;
\draw[fill=black!18,fill opacity=0.5]
(w5_2) -- (w5_3) -- (v5_3) -- (v5_2) -- cycle
;

\draw[fill=black!85,fill opacity=0.5]
(w0_0)
\foreach \y in {1,...,3}
{ -- (w0_\y)}
-- (t5) -- (s5) -- cycle
;

\draw[fill=black!16,fill opacity=0.5]
(v1_3)
\foreach \x in {2,...,5}
 { -- (v\x_3)}
-- (w5_3)
\foreach \x in {0,...,5}
 { -- (t\x)}
-- (w0_3) -- cycle
;

\end{scope}

\draw[]
(r0)
\foreach \x in {1,...,5}
{-- (r\x)}
;

\node[shift={(.3cm,.2cm)}] at (r5) {$\gamma$};

\draw[green!50!black]
(r1) -- ++(0,0,-.4)
(r2) -- ++(0,0,-.4)
(r3) -- ++(0,0,-.4)
(r4) -- ++(0,0,-.4)
;

\draw[very thick, red!90!black]
(w1_0) -- (w2_0) -- (w3_0) -- (w4_0)
;

\end{scope}
\end{scope}

&

\\[.5cm]

%
%
%

\begin{scope}[scale=1.5]
\begin{scope}[shift={(.15,0,-.6)}]

\input{common.tex}

\begin{scope}[opacity=0.2]

\path
(-.15,0,.6)
\foreach \x in {0,...,5}
{ +({50-\x*20}:1) coordinate (r\x) }
;

\path 
(v5_1) coordinate (b)
;
\foreach \x in {4,...,1}
{
\path 
(b) coordinate (b\x)
(v\x_1) coordinate (b)
;
}
\foreach \x in {4,...,1}
{
\path
(w\x_1) ++(-1,0,0) coordinate (a) 
(intersection of b\x--v\x_1 and w\x_1--a) coordinate (c\x)
;
\path
(w\x_1) ++(0,0,-1) coordinate (e)
(intersection of c\x--w\x_0 and w\x_1--e) 
coordinate (d\x)
;
}

\foreach \x in {0,...,5}
\path
(w\x_1) ++(0,0,-1) coordinate (e)
(w\x_0) ++(-1,0,0) coordinate (a)
(intersection of e--w\x_1 and a--w\x_0) coordinate (f\x)
;

\foreach \x in {0,...,5}
{
\path
(w\x_1) ++(0,0,1) coordinate (g\x)
;
}

\path
(w0_0) ++(-1,0,0) coordinate (h0)
++(0,0,.936) coordinate (h1)
(w5_0) ++(-1,0,.936) coordinate (h2)
;


\draw[very thin, black!30]
(w5_0) -- (s0) -- (t0) -- (w5_3)
(s4) -- (s5)
(t4) -- (t5)
(s0) -- (s1)
(t0) -- (t1)
;

\draw[very thin, black!30]
(s1) -- (s2) -- (s3) -- (s4)
(s1) -- (t1)
(s2) -- (t2)
(s3) -- (t3)
(s4) -- (t4)
;

\path 
(v5_0) coordinate (b);
\foreach \x in {5,...,2}
{
\path 
(b) coordinate (b\x)
(v\x_0) coordinate (a\x) coordinate (b);
}

\path
(v5_1) coordinate (b);
\foreach \x in {5,...,2}
{
  \pgfmathsetmacro\c{.09*\x}
  \definecolor{currentcolor}{hsb}{0, 0, \c}
  \draw[fill=currentcolor,fill opacity=0.5]
  (a\x) -- (v\x_1) -- (b\x) -- cycle;
  \path 
  (b) coordinate (b\x)
  (v\x_1) coordinate (a\x) coordinate (b);
}

\path
(v1_1) coordinate (a1)
(v2_1) coordinate (b1);

\foreach \y in {2,...,3}
{
 \path (v5_\y) coordinate (b); 
 \foreach \x in {5,...,1}
 {
  \pgfmathsetmacro\c{.09*\x+.11*\y}
  \definecolor{currentcolor}{hsb}{0, 0, \c}
  \draw[fill=currentcolor,fill opacity=0.5]
  (a\x) -- (v\x_\y) -- (b\x) -- cycle;
  \path 
  (b) coordinate (b\x)
  (v\x_\y) coordinate (a\x) coordinate (b);
 }
}

\foreach \x in {2,...,5}
{
 \path
 (v\x_0) coordinate (a\x)
 ;
}

\foreach \y in {1,...,3}
{
 \path
 (v1_\y) coordinate (b)
 ;
 \foreach \x in {2,...,5}
 {
  \pgfmathsetmacro\c{.09*\x+.11*\y+.04}
  \definecolor{currentcolor}{hsb}{0, 0, \c}
  \draw[fill=currentcolor,fill opacity=.5]
  (b) -- (a\x) -- (v\x_\y) coordinate (a\x) coordinate (b) -- cycle
  ;
 }
}

\draw[fill=black,fill opacity=0.5]
(w0_0) -- (w0_1) -- (v1_1) -- (w1_0) -- cycle
;
\draw[fill=black!80,fill opacity=0.5]
(w0_1) -- (w0_2) -- (v1_2) -- (v1_1) -- cycle
;
\draw[fill=black!67,fill opacity=0.5]
(w0_2) -- (w0_3) -- (v1_3) -- (v1_2) -- cycle
;

\draw[fill=black!47,fill opacity=0.5]
(w5_0) -- (w5_1) -- (v5_1) -- (w4_0) -- cycle
;
\draw[fill=black!34,fill opacity=0.5]
(w5_1) -- (w5_2) -- (v5_2) -- (v5_1) -- cycle
;
\draw[fill=black!18,fill opacity=0.5]
(w5_2) -- (w5_3) -- (v5_3) -- (v5_2) -- cycle
;

\draw[fill=black!85,fill opacity=0.5]
(w0_0)
\foreach \y in {1,...,3}
{ -- (w0_\y)}
-- (t5) -- (s5) -- cycle
;

\draw[fill=black!16,fill opacity=0.5]
(v1_3)
\foreach \x in {2,...,5}
 { -- (v\x_3)}
-- (w5_3)
\foreach \x in {0,...,5}
 { -- (t\x)}
-- (w0_3) -- cycle
;

\end{scope}

\draw[]
(r0)
\foreach \x in {1,...,5}
{-- (r\x)}
;

\node[shift={(.3cm,.2cm)}] at (r5) {$\gamma$};

\draw[green!50!black]
(d1) -- (r1)
(r2) -- (d2)
(r3) -- (d3)
(r4) -- (d4)
;

\draw[thick, blue]
(d1) -- (v2_1) -- (d2) -- (v3_1) -- (d3) -- (v4_1) -- (d4)
;

\end{scope}
\end{scope}

&

\begin{scope}[scale=1.5]
\begin{scope}[shift={(.15,0,-.6)}]

\input{common.tex}

\begin{scope}[opacity=0.2]

\path
(-.15,0,.6)
\foreach \x in {0,...,5}
{ +({50-\x*20}:1) coordinate (r\x) }
;

\path 
(v5_1) coordinate (b)
;
\foreach \x in {4,...,1}
{
\path 
(b) coordinate (b\x)
(v\x_1) coordinate (b)
;
}
\foreach \x in {4,...,1}
{
\path
(w\x_1) ++(-1,0,0) coordinate (a) 
(intersection of b\x--v\x_1 and w\x_1--a) coordinate (c\x)
;
\path
(w\x_1) ++(0,0,-1) coordinate (e)
(intersection of c\x--w\x_0 and w\x_1--e) 
coordinate (d\x)
;
}

\foreach \x in {0,...,5}
\path
(w\x_1) ++(0,0,-1) coordinate (e)
(w\x_0) ++(-1,0,0) coordinate (a)
(intersection of e--w\x_1 and a--w\x_0) coordinate (f\x)
;

\foreach \x in {0,...,5}
{
\path
(w\x_1) ++(0,0,1) coordinate (g\x)
;
}

\path
(w0_1) coordinate (d0)
(w5_1) coordinate (d5)
(d1) coordinate (a)
(f1) coordinate (b)
;
\foreach \x in {2,...,4}
{
  \pgfmathsetmacro\c{.2*\x}
  \definecolor{currentcolor}{hsb}{.66666, .4, \c}
  \fill[currentcolor,fill opacity=0.5]
  (a) -- (v\x_1) -- (d\x) -- (f\x) -- (b) -- cycle;
  \path 
  (d\x) coordinate (a)
  (f\x) coordinate (b)
  ;
}

\end{scope}

\node[shift={(.3cm,.2cm)}] at (r5) {$\gamma$};

\draw[]
(r0)
\foreach \x in {1,...,5}
{-- (r\x)}
;

\draw[green!50!black]
(r1) -- (f1)
(r2) -- (f2)
(r3) -- (f3)
(r4) -- (f4)
;

\draw[thick, blue]
(d1) -- (v2_1) -- (d2) -- (v3_1) -- (d3) -- (v4_1) -- (d4)
;

\end{scope}
\end{scope}

\\[-.5cm]

\begin{scope}[scale=1.5]

\path
(10:-1) ++(0,0,.936) ++(30:-1) coordinate (o)
;

\begin{scope}[shift={(o)}]

\input{common.tex}

\path
(t2) ++(30:-1)
\foreach \x in {0,...,5}
{ +({50-\x*20}:1) coordinate (r\x) }
;


\draw[thick, green!50!gray]
(t1) -- (s1) -- (s2) -- (t2)
;

\begin{scope}[opacity=0.2]

\draw[very thin, black!30]
(w5_0) -- (s0) -- (t0) -- (w5_3)
(s4) -- (s5)
(t4) -- (t5)
(s0) -- (s1)
(t0) -- (t1)
;

\draw[very thin, black!30]
(s1) -- (s2) -- (s3) -- (s4)
(s1) -- (t1)
(s2) -- (t2)
(s3) -- (t3)
(s4) -- (t4)
;

\path 
(v5_0) coordinate (b);
\foreach \x in {5,...,2}
{
\path 
(b) coordinate (b\x)
(v\x_0) coordinate (a\x) coordinate (b);
}

\path
(v5_1) coordinate (b);
\foreach \x in {5,...,2}
{
  \pgfmathsetmacro\c{.09*\x}
  \definecolor{currentcolor}{hsb}{0, 0, \c}
  \draw[fill=currentcolor,fill opacity=0.5]
  (a\x) -- (v\x_1) -- (b\x) -- cycle;
  \path 
  (b) coordinate (b\x)
  (v\x_1) coordinate (a\x) coordinate (b);
}

\path
(v1_1) coordinate (a1)
(v2_1) coordinate (b1);

\foreach \y in {2,...,3}
{
 \path (v5_\y) coordinate (b); 
 \foreach \x in {5,...,1}
 {
  \pgfmathsetmacro\c{.09*\x+.11*\y}
  \definecolor{currentcolor}{hsb}{0, 0, \c}
  \draw[fill=currentcolor,fill opacity=0.5]
  (a\x) -- (v\x_\y) -- (b\x) -- cycle;
  \path 
  (b) coordinate (b\x)
  (v\x_\y) coordinate (a\x) coordinate (b);
 }
}

\foreach \x in {2,...,5}
{
 \path
 (v\x_0) coordinate (a\x)
 ;
}

\foreach \y in {1,...,3}
{
 \path
 (v1_\y) coordinate (b)
 ;
 \foreach \x in {2,...,5}
 {
  \pgfmathsetmacro\c{.09*\x+.11*\y+.04}
  \definecolor{currentcolor}{hsb}{0, 0, \c}
  \draw[fill=currentcolor,fill opacity=.5]
  (b) -- (a\x) -- (v\x_\y) coordinate (a\x) coordinate (b) -- cycle
  ;
 }
}

\draw[fill=black,fill opacity=0.5]
(w0_0) -- (w0_1) -- (v1_1) -- (w1_0) -- cycle
;
\draw[fill=black!80,fill opacity=0.5]
(w0_1) -- (w0_2) -- (v1_2) -- (v1_1) -- cycle
;
\draw[fill=black!67,fill opacity=0.5]
(w0_2) -- (w0_3) -- (v1_3) -- (v1_2) -- cycle
;

\draw[fill=black!47,fill opacity=0.5]
(w5_0) -- (w5_1) -- (v5_1) -- (w4_0) -- cycle
;
\draw[fill=black!34,fill opacity=0.5]
(w5_1) -- (w5_2) -- (v5_2) -- (v5_1) -- cycle
;
\draw[fill=black!18,fill opacity=0.5]
(w5_2) -- (w5_3) -- (v5_3) -- (v5_2) -- cycle
;

\draw[fill=black!85,fill opacity=0.5]
(w0_0)
\foreach \y in {1,...,3}
{ -- (w0_\y)}
-- (t5) -- (s5) -- cycle
;

\draw[fill=black!16,fill opacity=0.5]
(v1_3)
\foreach \x in {2,...,5}
 { -- (v\x_3)}
-- (w5_3)
\foreach \x in {0,...,5}
 { -- (t\x)}
-- (w0_3) -- cycle
;

\end{scope}

\draw[]
(r0)
\foreach \x in {1,...,5}
{-- (r\x)}
;

\node[shift={(.3cm,.2cm)}] at (r5) {$\gamma$};

\draw[ultra thick, green!50!black]
(r1) -- (r2)
;

\end{scope}
\end{scope}

&

\begin{scope}[scale=1.5]

\path
(10:-1) ++(0,0,.936) ++(30:-1) coordinate (o)
;

\begin{scope}[shift={(o)}]

\input{common.tex}

\path
(t2) ++(30:-1)
\foreach \x in {0,...,5}
{ +({50-\x*20}:1) coordinate (r\x) }
;

\begin{scope}[opacity=0.2]

\path 
(v5_0) coordinate (b);
\foreach \x in {5,...,2}
{
\path 
(b) coordinate (b\x)
(v\x_0) coordinate (a\x) coordinate (b);
}

\path
(v5_1) coordinate (b);
\foreach \x in {5,...,2}
{
\draw[very thin, black!30]
  (a\x) -- (v\x_1) -- (b\x) -- cycle;
  \path 
  (b) coordinate (b\x)
  (v\x_1) coordinate (a\x) coordinate (b);
}

\path
(v1_1) coordinate (a1)
(v2_1) coordinate (b1);

\foreach \y in {2,...,3}
{
 \path (v5_\y) coordinate (b); 
 \foreach \x in {5,...,1}
 {
\draw[very thin, black!30]
  (a\x) -- (v\x_\y) -- (b\x) -- cycle;
  \path 
  (b) coordinate (b\x)
  (v\x_\y) coordinate (a\x) coordinate (b);
 }
}

\draw[very thin, black!30]
(v1_0) -- (v2_0) -- (v1_1)
(v1_3)
\foreach \x in {2,...,5}
 { -- (v\x_3)}
;

\draw[very thin, black!30]
(w0_0)
\foreach \y in {1,...,3}
{ -- (w0_\y)}
;

\draw[very thin, black!30]
(w5_0)
\foreach \y in {1,...,3}
{ -- (w5_\y)}
;

\foreach \y in {0,...,3}
{
\draw[very thin, black!30]
(w0_\y) -- (v1_\y)
(w5_\y) -- (v5_\y)
;
}

\draw[very thin, black!30]
(t5) -- (w0_3)
;

\path
(s5) coordinate (a)
(t5) coordinate (b)
;
\foreach \x in {4,...,0}
{
  \pgfmathsetmacro\c{.6-.09*\x}
  \definecolor{currentcolor}{hsb}{0, 0, \c}
  \draw[fill=currentcolor,fill opacity=0.5]
  (a) -- (s\x) -- (t\x) -- (b) -- cycle;
  \path 
  (s\x) coordinate (a)
  (t\x) coordinate (b)
  ;
}

\draw[fill=black, fill opacity=0.5]
(w0_0)
\foreach \x in {1,...,5}
 { -- (w\x_0)}
\foreach \x in {0,...,5}
 { -- (s\x)}
 -- cycle
;

\draw[fill=black!20,fill opacity=0.5]
(w5_0)
\foreach \y in {1,...,3}
{ -- (w5_\y)}
-- (t0) -- (s0) -- cycle
;

\end{scope}

\draw[]
(r0)
\foreach \x in {1,...,5}
{-- (r\x)}
;

\node[shift={(.3cm,.2cm)}] at (r5) {$\gamma$};

\draw[very thick, green!50!black, fill=green!50!black,fill opacity=0.2]
(t1) -- (s1) -- (s2) -- (t2) -- cycle
;

\end{scope}
\end{scope}

\\
};

\end{yxzcoords}
\end{tikzpicture}

%% file: universal.tex
\begin{tikzpicture}

\begin{yxzcoords}

\fill[black!10]
(-9.5cm,-1cm) coordinate (o) circle (1cm)
;

\path
(o)
+(0,0,.4) node (z) {\ns$z$}
+(.5,0,0) node (y) {\ns$y$}
+(0,.4,0) node (x) {\ns$x$}
;

\draw[thin,->] (o) -- (x);
\draw[thin,->] (o) -- (y);
\draw[thin,->] (o) -- (z);

\end{yxzcoords}

\path 
(20:2.6667 and 1.3333) coordinate (x1)
(-70:2 and 1) coordinate (x2)
(90:1.7321) coordinate (x3)
;

\begin{scope}[x={(x2)}, y={(x1)}, z={(x3)},scale=3]

\path
\foreach \x in {0,...,12}
{ ({-(2/3)^\x*90}:1) coordinate (w\x) };

\path
\foreach \x in {0,...,12}
{
 \foreach \y in {0,...,18}
 {
  (w\x) ++({((5/6)^\y-1)},0,{(1-(1/2)^\y)}) coordinate (w\x_\y) 
 }
 (w\x) ++(-1,0,1) coordinate (w\x_y)
}
\foreach \y in {0,...,18}
{
 (w0_\y) ++(1,1,0) coordinate (wx_\y)
}
(0,0,1) coordinate (wx_y);

\foreach \y in {0,...,18}
{
 \path (w0_\y) coordinate (a); 
 \foreach \x in {1,...,12}
 {
  \path
  (w\x_\y) coordinate (b)
  ($(b)!{(3/5)^\y}!(a)$) coordinate (v\x_\y)
  (b) coordinate (a);
 }
}

\path
\foreach \x in {0,...,12}
{ (-1,-1,0) ++({-(2/3)^\x*90}:-1) coordinate (s\x)
  ++(0,0,1) coordinate (t\x)}
;

\path
(-2,-1,0) coordinate (sx)
(-2,-1,1) coordinate (tx)
;

\draw[ultra thin, black!20]
(wx_0) -- (s0)
\foreach \x in {1,...,12}
 { -- (s\x) }
-- (sx) 
;


\foreach \x in {0,...,12}
\draw[ultra thin, black!20]
 (s\x) -- (t\x)
;

\draw[ultra thin, black!20, fill=black!20]
(s12) -- (sx) -- (tx) -- (t12) -- cycle
;

\path
(w0_0) coordinate (a)
(v2_0) coordinate (b)
;
\foreach \y in {1,...,18}
{
\pgfmathsetmacro\c{.07*\y+.03}
\definecolor{currentcolor}{hsb}{0, 0, \c}
\draw[ultra thin, black!50!currentcolor,fill=currentcolor,fill opacity=.7]
(a) -- (w0_\y) -- (v1_\y) -- (b) -- cycle
;
\path
(w0_\y) coordinate (a)
(v1_\y) coordinate (b)
;
}

\path 
(v12_0) coordinate (b);
\foreach \x in {11,...,2}
{
\path 
(b) coordinate (b\x)
(v\x_0) coordinate (a\x) coordinate (b);
}

\path
(v12_1) coordinate (b);
\foreach \x in {11,...,2}
{
  \pgfmathsetmacro\c{.09*\x+.06}
  \definecolor{currentcolor}{hsb}{0, 0, \c}
  \draw[ultra thin, black!50!currentcolor,fill=currentcolor,fill opacity=.7]
  (a\x) -- (v\x_1) -- (b\x) -- cycle;
  \path 
  (b) coordinate (b\x)
  (v\x_1) coordinate (a\x) coordinate (b);
}

\path
(v1_1) coordinate (a1)
(v2_1) coordinate (b1);

\foreach \y in {2,...,18}
{
 \path (v12_\y) coordinate (b); 
 \foreach \x in {11,...,1}
 {
  \pgfmathsetmacro\c{.09*\x+.07*\y}
  \definecolor{currentcolor}{hsb}{0, 0, \c}
  \draw[ultra thin, black!50!currentcolor,fill=currentcolor,fill opacity=.7]
  (a\x) -- (v\x_\y) -- (b\x) -- cycle;
  \path 
  (b) coordinate (b\x)
  (v\x_\y) coordinate (a\x) coordinate (b);
 }
}

\foreach \x in {2,...,12}
{
 \path
 (v\x_0) coordinate (a\x)
 ;
}

\foreach \y in {1,...,18}
{
 \path
 (v1_\y) coordinate (b)
 ;
 \foreach \x in {2,...,12}
 {
  \pgfmathsetmacro\c{.09*\x+.07*\y+.06}
  \definecolor{currentcolor}{hsb}{0, 0, \c}
  \draw[ultra thin, black!50!currentcolor,fill=currentcolor,fill opacity=.7]
  (b) -- (a\x) -- (v\x_\y) coordinate (a\x) coordinate (b) -- cycle
  ;
 }
}

\shade[draw, ultra thin, black!55, top color= black!55, middle color= black!50, bottom color= black!40]
(v12_0)
\foreach \y in {0,...,18}
{ -- (wx_\y) }
-- (wx_y)
\foreach \x in {12,...,0}
{ -- (w\x_y) }
-- (w0_18) 
\foreach \x in {1,...,12}
{ -- (v\x_18) }
\foreach \y in {18,...,0}
{ -- (v12_\y) }
;

\draw[ultra thin, black!60, fill=black!20,fill opacity=0.5]
(w0_y) 
\foreach \x in {1,...,12}
{ -- (w\x_y) }
-- (wx_y) -- (t0)
\foreach \x in {1,...,12}
 { -- (t\x) }
-- (tx) -- cycle
;

\draw[ultra thin, black!90, fill=black,fill opacity=0.5]
(w0_0)
\foreach \y in {1,...,18}
{ -- (w0_\y) }
-- (w0_y) -- (tx) -- (sx) -- cycle
;

\end{scope}

\end{tikzpicture}